\documentclass[12pt]{article}
\input epsf

\newcommand{\qed}{\rule[-0.2ex]{0.3em}{1.4ex}}

{\par\vspace{0.5ex}\noindent{\bf Proof:}\hspace{0.5em}}%
{\nopagebreak
\strut\nopagebreak
\hspace{\fill}\qed\par\medskip\noindent}
\newenvironment{proof}% lowercase version...
{\par\vspace{0.5ex}\noindent{\bf Proof:}\hspace{0.5em}}%
{\nopagebreak
\strut\nopagebreak
\hspace{\fill}\qed\par\medskip\noindent}
{\par\vspace{0.5ex}\noindent{\bf Proof Attempt:}\hspace{0.5em}}%
{\nopagebreak
\strut\nopagebreak
\hspace{\fill}\qed\par\medskip\noindent}

\usepackage[pdftex]{graphicx}
\usepackage{amssymb}
\usepackage{amsmath}

\usepackage{textcomp}
\usepackage{booktabs}

\newenvironment{myindentpar}[1]%
  {\begin{list}{}%
          {\setlength{\leftmargin}{#1}}%
          \item[]%
  }
{\end{list}}

\newtheorem{theorem}{Theorem}

\newtheorem{lemma}{Lemma}
\newtheorem{fact}{Fact}
\newtheorem{definition}{Definition}
\newtheorem{corollary}{Corollary}

\newtheorem{convention}{Convention}

\newcommand{\IR}{{\rm\hbox{I\kern-.15em R}}}
\newcommand{\reals}{{\rm\hbox{I\kern-.15em R}}}
\newcommand{\IN}{{\rm\hbox{I\kern-.15em N}}}
\newcommand{\IZ}{{\sf\hbox{Z\kern-.40em Z}}}

\newcommand{\E}{{\rm\hbox{I\kern-.15em E}}}

\newcommand{\Prob}{\mbox{\sf \hbox{I\kern-.15em P}}}

\newcommand{\lc}{\mbox{\bf \em lc}}

\begin{document}
\title{{\bf Efficient Exact Paths For Dyck and semi-Dyck Labeled Path Reachability\thanks{This current paper uses standard definitions of Dyck and semi-Dyck languages.
The author's earlier abstracts reversed the Dyck and semi-Dyck definitions. 
An extended abstract of this paper is in the proceedings of the UEMCON 2017 conference, see~\cite{Bradford2017}.} }}
\author{
Phillip G. Bradford\thanks{\sc
Department of Computer Science and Engineering,
University of Connecticut,
Stamford, CT, USA. \ \
{\tt phillip.bradford@uconn.edu}, \ {\tt phillip.g.bradford@gmail.com}
}
}

\date{\small\today}

\maketitle
\begin{abstract}
The {\em exact path length problem} is to determine if there is a path of a given fixed cost between two vertices.
This paper focuses on the exact path problem for costs $-1,0$ or $+1$ between all pairs of vertices in an edge-weighted digraph.
The edge weights are from $\{ -1, +1 \}$.
In this case, this paper gives an $\widetilde{O}(n^{\omega})$ exact path solution.
Here $\omega$ is the best exponent for matrix multiplication and $\widetilde{O}$ is the asymptotic upper-bound mod polylog factors.

Variations of this algorithm determine which pairs of digraph nodes have Dyck or semi-Dyck labeled paths between them, 
assuming two parenthesis.
Therefore, determining digraph reachability for Dyck or semi-Dyck labeled paths costs $\widetilde{O}(n^{\omega})$.
A path label is made by concatenating all symbols along the path's edges.

The exact path length problem has many applications.
These applications include the labeled path problems given here, which in turn, also have numerous applications.
\end{abstract}

\section{Introduction}
% no \IEEEPARstart
%
%

Shortest path algorithms are a great success.
Many people use them and many vehicles are equipped with them.
Determining path reachability is also important.
Path reachability is often computed using transitive closure.

This paper efficiently solves the $0$ and $\pm 1$ exact path length problem for digraphs whose edges have weights from
$\{ -1, +1 \}$.

Context-free language constrained graph problems are fundamental to a plethora of challenges.  
This paper gives algorithms for determining Dyck (semi-Dyck) constrained paths on digraphs based on the exact path problem.
Dyck and semi-Dyck context-free languages are important.
A central application for the exact path problem is for determining Dyck and semi-Dyck constrained paths in digraphs.
Here these languages have a single parenthesis type.

\begin{definition}[Exact path length problem~\cite{NU}]
{\sf
Consider an integer edge weighted digraph $G$.
Given an integer $\kappa$, the {\em {\bf EPL} (exact path length problem)}
is to determine whether there is a path between a given pair vertices costing exactly~$\kappa$.
}
\label{Exact_Path}
\end{definition}

Nyk\"anen and Ukkonen~\cite{NU} show the general EPL is ${\cal NP}$-Complete.
They also give a pseudo-polynomial algorithm for the EPL\@. 
The current paper uses a special case of the EPL where $\kappa \in \{ -1, 0, +1 \}$ and edge costs are from the set $\{ -1, +1 \}$. 

Given these restricted edge costs, and for $\kappa \neq 0$,  applying Nyk\"anen and Ukkonen's algorithm costs
$O(n^3 + n^{\omega} \log |\kappa|)$ time\footnote{All logs are base 2 except where specified otherwise.}, see~\cite{NU}.
For $\kappa =0$, their algorithm costs $O(n^3)$.

Solving this Dyck (semi-Dyck) labeled path problem is interesting due to the close relationship between transitive closure,
Boolean and algebraic matrix multiplication, and context-free grammar recognition.
For example, Lee~\cite{Lee:02a} gives an equivalence between Context-free parsers and Boolean
matrix multiplication algorithms.

\subsection{Semi-Dyck and Dyck Constrained Graphs} 
\label{subsec:semi-dyck-and-dyck-constrained-graphs}

Dyck and semi-Dyck languages are parenthesis languages. Dyck or semi-Dyck languages with two
parenthesis symbols and $n$ total parentheses can be parsed in $O(n)$ time and space.
However, efficiently computing Dyck (and semi-Dyck) constrained reachability on digraphs seems more challenging.

Let ${\cal D}$ be a Dyck language  of one open-parentheses symbol ${\bf a}$ and  one close-parentheses symbol ${\bf a^{-1}}$.  
A sentence $w \in {\cal D}$ iff $w$ can be
reduced using right-inverse reduction, e.g. ${\bf a \, a^{-1}} = \epsilon$, to the empty string $\epsilon$.  
The Dyck language ${\cal D}$ is derivable from the grammar:
\begin{eqnarray*}
{\cal D} & \Longrightarrow & 
\epsilon \ |\ {\cal D} \; {\cal D} \ |\ {\bf a} \; {\cal D} \; {\bf a^{-1}}.
%\label{eqn:Dyck-grammar}
\end{eqnarray*}

Semi-Dyck languages allow reductions using both right-inverses and left-inverses ${\bf a a^{-1}} = {\bf a^{-1} a} = \epsilon$.
They are derivable from the grammar:
\begin{eqnarray*}
{\cal S} & \Longrightarrow &
\epsilon \ |\ {\cal S}{\cal S} \ |\ {\bf a} {\cal S} {\bf a^{-1}} \ | \ {\bf a^{-1}} {\cal S} {\bf a}.
%\label{eqn:semi-Dyck-grammar}
\end{eqnarray*}

Dyck languages generate all strings of balanced parenthesizations.
Semi-Dyck languages generate all strings of equal numbers of matching symbols.

%
%

%
% adjuct the second column for IEEE Copyright
%

The next definition is similar to one in~\cite{BJM}.

\begin{definition}[Labeled Directed Graph]
{\sf
A {\em labeled directed graph} (LDG) is a multigraph $G = (\Sigma, \, V, \, E_1)$
consisting of a set $V$ of vertices and a set $E_1 \subseteq V \times V \times \Sigma$ of labeled and directed edges.

The set $\Sigma$ contains a grammar's terminals. If the grammar is Dyck (semi-Dyck), then $\Sigma$ is said to be Dyck (semi-Dyck).
}
\label{LDGs}
\end{definition}

Given Definition~\ref{LDGs}, restrict cycles to having no repeated edges.
LDGs are multigraphs.
All LDG edges are augmented with {\em label-costs}.
So each edge $e$ in $G$ has a label $l(e)$ and a label-cost $\lc(e)$.
The label-cost function is,
\begin{eqnarray*}
\lc(e)
& = &
\left\{
\begin{array}{ll}
-1 	& \mbox{ if } l(e) = a^{-1} \\
+1  	& \mbox{ if } l(e) = a.
\end{array}
\right.
\end{eqnarray*}

A $+1$ edge and a $-1$ edge may be joined to form a new $0$ label-cost edge for computing an exact path. 
After some processing, say such a new $0$ label-cost edge $e$ is created.
Then, the label-cost function extends so $\lc(e) = 0$.
This new $0$ label-cost edge is added to an augmented edge set in the LDG.
Also, $\pm 1$ and $0$ label-cost edges may be extended by adjoining $0$ label-cost edges.

The label-costs or costs are written above edges such as $e = i \stackrel{w}{\longrightarrow} j$.
Therefore, in general $w \in \{ -1,0,+1 \}$, but at the start of our algorithms assume $w \in \{ -1, +1 \}$.

\subsection{Previous Work}

Greenlaw, Hoover, and Ruzzo~\cite{GHR} discuss several formal-language based 
reachability problems. See also Afrati and Papadimitriou~\cite{Afrati:1987:PCS:28659.28682}, Reps~\cite{Reps:1996:SNI:2697697.2697789},
and Ullman and Van Gelder~\cite{4568236}.
For example, the {\em LGAP (labeled graph accessibility problem)}~\cite{GHR} is a Dyck language with constrained reachability 
problem on a directed graph $G$ that is ${\cal P}$-complete when $|\Sigma| \geq 4$.
Yannakakis~\cite[p. 237]{Yannakakis90} points out that Valiant's Boolean matrix multiplication context-free word recognition
algorithm determines single-source labeled path reachability in DAGs.
This means there is an algorithm  costing $O(n^{\omega} \log n)$ for finding context-free labeled and unweighted paths in DAGs
with $n$ vertices, where $\omega$ is the best exponent for $n\times n$ matrix multiplication.
Very efficient matrix multiplication algorithms include
results of Coppersmith and Winograd~\cite{COPPERSMITH1990}; Stothers~\cite{Stothers:2010}; Williams~\cite{Williams:2012}; and Le Gall~\cite{LeGall:2014}.
Currently, the best exponent of square matrix multiplication is $\omega < 2.373$.

Melski and Reps~\cite{MR} give an $O(n^3|S|^3)$ context-free language reachability algorithm.
Where $S$ is the set of terminals and non-terminals for the input grammar.
Barrett, Jacob, and Marathe~\cite{BJM} give an $O(n^3|R||N|)$ algorithm for finding the all-pairs shortest paths in 
context free grammar constrained path problems. Here $R$ is the set of rules and $N$ is the set of
non-terminals in Chomsky normal form\@.
This algorithm does not compute shortest paths with negative edge weights.

Alon, Galil, and Margalit~\cite{AGM} give efficient algorithms for shortest paths on digraphs with 
edge weights from $\{ -1, 0, +1 \}$ costing $\widetilde{O}(n^{\nu})$ where $\nu = \frac{3+\omega}{2}$.
See also Takaoka~\cite{Takaoka1998}.
A part of Alon, et al.'s algorithm finds zero length directed paths and uses them as short-cuts.
It may be possible to extract their $0$ length path algorithm for short-cuts as the basis of our work.
Nonetheless, Alon, et al.'s directed graph shortest-path algorithm takes $\widetilde{O}(n^{2.687})$ time when $\omega < 2.373$.

Galil and Margalit~\cite{GM} extend the results of Alon, et al.~\cite{AGM} and
integrate the shortest path distance and shortest path problem.
Zwick~\cite{Zwick:2002} gives more efficient all-pairs shortest path and path distance algorithms. 
Zwick's shortest path's cost is better than $\widetilde{O}(n^{2.575})$, since $\omega < 2.373$.
Our algorithm does not solve shortest-path problems.

Building on Barrett,~et al.,
Bradford and Thomas~\cite{BT}  give a more efficient context free label constrained shortest path algorithm
for graphs with positive and negative edge weights whose unlabeled versions have no negative cycles.
Barrett,~et al.'s algorithm for finding context free label constrained shortest paths with positive and negative edge weights costs $O(n^5 |N|^2 |R|^2)$.
Bradford~\cite{B-2007} gives a solution to a quickest-path problem for context-free grammars applied to cryptographic routing.
Bradford and Choppella~\cite{BC} use a special-case of Nyk\"anen and Ukkonen's sign-closure algorithm for DAGs with initial edge costs
from $\{ -1, +1 \}$, see also Khamespanah, Khosravi, and Sirjani~\cite{KKS}. Further Bradford and Choppella~\cite{BC}
find actual minimum-cost point-to-point Dyck paths in DAGs.
Ward, Wiegand, and Bradford~\cite{WWB} give a distributed context-free labeled graph shortest path algorithm
also based on~\cite{BJM}.
Ward and Wiegand~\cite{WW} analyze the complexity of wireless routing metrics as labeled path problems.

Chaudhuri~\cite{C} gives an $O(n^3/\log n)$ algorithm for context-free language reachability using an 
important dynamic-programming speedup method by Rytter~\cite{Rytter}.

After preprocessing a bidirected tree, Yuan and Eugster~\cite{Yuan2009} give a $O(|V| \log |V|)$ algorithm for finding Dyck reachability in bidirected trees
in $O(1)$ per query.

Zhang, Lyu, Yuan, Hao and Su~\cite{Zhang:2013:FAD:2499370.2462159} improve on  Yuan and Eugster's bidirected tree algorithm.
Zhang, et al.~\cite{Zhang:2013:FAD:2499370.2462159} also give an $O(|V| + |E| \log |E|)$ algorithm 
for determining Dyck reachability for bidirected digraphs.
Each Dyck labeled edge in a bidirected digraph has a mirror edge going in the opposite direction and with
a complimentary label.
See also~\cite{Zhang:2017:CDA:3093333.3009848,Hollingum2015}.

Khamespanah, Khosravi, and Sirjani~\cite{KKS} use Nyk\"anen and Ukkonen~\cite{NU}'s exact path algorithm 
to improve their model checking algorithms for timed actors in distributed systems.
They apply the pseudo-polynomial algorithm's $O(n^2)$ path relaxation cost,
while accepting the pre-processing costs of $O(n^3)$.
Our results break through this $O(n^3)$ barrier giving a $\widetilde{O}(n^{\omega})$ algorithm.
In general, Melski and Reps~\cite{MR} discuss the ``$O(n^3)$ bottleneck''
for context-free program analysis.
Our results solidly break through this bottleneck for the Dyck and semi-Dyck cases.

Dyck and semi-Dyck languages are also applied to data streaming, see Chakrabarti, Cormode, Kondapally and McGregor~\cite{DBLP:journals/siamcomp/ChakrabartiCKM13}.
In addition, Tang, et al.~\cite{10.1007/978-3-662-54434-1_33} apply Dyck-CLF reachability to library summarization.
Likewise, there are applications to database path queries, see Grahne, Thomo and Wadge~\cite{DBLP:journals/fuin/GrahneTW08}.
Choppella and Haynes give an equivalence between unification graphs and Dyck path
reachability problems in digraphs~\cite{CH05}.  

\subsection{Structure of this paper}

Section~\ref{sec:transformation} gives the foundations for the rest of the paper.
Section~\ref{sec:SP-distances} shows how to find efficient exact $0$ cost paths for graphs labeled with flat Dyck and semi-Dyck grammars.
Subsection~\ref{AGMY-encoding} leverages Alon, Galil, Margalit's algebraic matrix encoding~\cite{AGM} for our solution. See also Yuval~\cite{Yuval}.
Where subsection~\ref{flat-grammars} gives an exact $0$ cost path solution for flat Dyck and flat semi-Dyck grammars.

Grid graphs for exact $0$ Dyck and semi-Dyck paths are given in section~\ref{sec:grid-graphs}.
This section has a number of definitions directly in the text. This is to simplify the flow.
Finally, subsection~\ref{general-case-subsection} concludes section~\ref{sec:grid-graphs} giving the general exact $0$ path solution.

Section~\ref{sec:plus-minus-one} extends the exact $0$ cost reachability to exact $\pm 1$ reachability.
\section{Dyck Path Reachability Problem}
\label{sec:transformation}

Direct application of standard shortest path~\cite{CLRS} and transitive closure~\cite{CLRS,pwp} algorithms to LDGs does not
seem to determine $0$ reachability. 
In our instantiation of this challenge, such shortest paths use $-1$ edge weights or label-costs. 
That is, some paths may be negative. 
Indeed, shortest path algorithms gravitate towards negative paths.
For this reason, intuitively shortest path algorithms not directly applicable.

This paper converts its edge labels to label-costs from $\{ -1, 0, +1 \}$.
Therefore, rather than referring to labeled-costs, this paper just discusses edge costs.
These costs are generally restricted to $\{ -1, 0, +1 \}$.

Next are definitions for sign-closure graphs from Nyk\"anen and Ukkonen~\cite{NU}.
They also define the function ${\bf sgn}$, for $w \in \{ -1, 0, +1 \}$ so that ${\bf sgn}(w) = w$.
For any LDG $G = (\Sigma, V,E_1)$, let $M$ be a label-cost bound,
\begin{eqnarray*}
M & = & \max \{ \, |\lc(e)| : e \in E_1 \, \}.
\end{eqnarray*}

Throughout this paper, $M = 1$.

\begin{definition}[Nyk\"anen and Ukkonen~\cite{NU}]
{\sf
Consider a digraph $G =(V,E)$. The {\em sign-closure} of $E(G)$ is {\bf unsign}$(G)$ which starts with
$E(\mbox{\bf unsign}(G)) \leftarrow E(G)$, and then apply the rule:
\begin{tabbing}
text\= text \= text \= text \= text \= text \kill
\> {\bf if} $i \stackrel{v}{\longrightarrow} k \stackrel{w}{\longrightarrow} j \in E({\bf unsign}(G))$ and ${\bf sgn}(v) \neq {\bf sgn}(w)$\\
\> {\bf then} put $i \stackrel{v+w}{\longrightarrow} j$ in $E({\bf unsign}(G))$,
\end{tabbing}
until it no longer applies.
}
\label{Def:Closure}
\end{definition}

Applying Nyk\"anen and Ukkonen's sign-closure algorithm finds semi-Dyck paths in an LDG.
This relates semi-Dyck paths to transitive closure.

Changing the if-statement in Definition~\ref{Def:Closure} as follows gives {\em Dyck} sign-closure.
Given a LDG $G$, its Dyck sign-closure is ${\bf unsign}^{\geq}(G)$.
To get the Dyck sign closure of a graph $G$, apply the rule

\begin{tabbing}
text\= text \= text \= text \= text \= text \kill
\> {\bf if} $i \stackrel{v}{\longrightarrow} k \stackrel{w}{\longrightarrow} j \in E({\bf unsign}^{\geq}(G))$ and ${\bf sgn}(v) \neq {\bf sgn}(w)$ and $v \neq -1$\\
\> {\bf then} put $i \stackrel{v+w}{\longrightarrow} j$ in $E({\bf unsign}^{\geq}(G))$,
\end{tabbing}
until it no longer applies.

Nyk\"anen and Ukkonen show the sign-closure graph problem is ${\cal NP}$-Complete.
Nonetheless, Nyk\"anen and Ukkonen give a $O(M^2 n^3)$ time pseudo-polynomial algorithm for computing a sign-closure graph. 
This pseudo-polynomial algorithm runs in polynomial time for edge costs restricted to $\{ -1,0,+1 \}$ since $M=1$.
In particular, when $M=1$, computing a sign-closure graph costs $O(n^3)$ by~\cite{NU}.
We improve the cost to $\widetilde{O}(n^{\omega})$.

The basic result of the next lemma is mentioned in the proof of Theorem 5 in Nyk\"anen and Ukkonen~\cite{NU}.
Their Theorem 5 assumes their $O(M^2 n^3)$ sign-closure algorithm.
Nyk\"anen and Ukkonen were not discussing Dyck or semi-Dyck languages, but in our context, their result is as follows.

\begin{lemma}[Nyk\"anen and Ukkonen~\cite{NU}]
{\sf
Consider an LDG $G = (\Sigma,V,E_1)$ where $\Sigma$ is Dyck (semi-Dyck) and $|\Sigma|= 2$.
In computing a sign-closure with edge costs from $\{ -1, +1 \}$, then new edges added to $E(H)$ 
may be limited to costs from $\{ \ -1, 0, +1 \ \}$.
}
\label{ZSC-Gives-Paths}
\end{lemma}

The case of Dyck languages follows since Dyck languages are also semi-Dyck languages.
Given an LDG $G$, a {\em $0$ cost} edge (path) is an edge (path) in ${\bf unsign}(G)$ or ${\bf unsign}^{\geq}(G)$.
A $0$ cost edge has label-cost computed to be~$0$
and a $0$ cost path has total cost~$0$.

Zero cost paths are semi-Dyck paths in $G$.
A proof of the next lemma follows since semi-Dyck paths along $\pm 1$ edges have equal numbers of $+1$ and $-1$ values.
See also~\cite{BC}.

\begin{lemma} 
{\sf
Consider the LDG $G = (\Sigma, V,E_1)$ where $\Sigma$ is semi-Dyck and $|\Sigma|=2$, then
$G$ has a semi-Dyck path between $i$ and $j$ iff in $G$ there is a $0$ cost path between $i$ and $j$.
}
\label{DYCK}
\end{lemma}

%\vspace{0.25in}

The next definition is well-known.

\begin{definition}[Non-negative prefix sum]
{\sf
Suppose $G$ is a LDG with a simple weighted path $p$ from $i_0$ to $i_{t+1}$:
	\begin{eqnarray*}
p \ \ = \ \ i_0 \stackrel{v_1}{\longrightarrow} i_1 \stackrel{v_2}{\longrightarrow} \cdots \stackrel{v_{t}}{\longrightarrow} i_{t} \stackrel{v_{t+1}}{\longrightarrow} i_{t+1}.
	\end{eqnarray*}
then node $i_k$ has prefix sum $v_1 + \cdots + v_{k}$ from $p=i_0$ to $i_k$ along $p$ for $k: t+1 \geq k \geq 1$. 
The prefix sum for $i_0$ is~$0$.
In a path $p$, if $p$'s prefix sums for all $0$ cost subpaths are non-negative, then the path $p$ has a non-negative prefix sum.
}
\label{Def:PrefixSum}
\end{definition}

A proof of the next lemma follows a proof of Lemma~\ref{DYCK}, see also~\cite{BC,4568236}.

\begin{lemma}
{\sf
Consider the LDG $G = (\Sigma, V,E_1)$ where $\Sigma$ is Dyck and $|\Sigma|=2$,
then $G$ has a Dyck path between $i$ and $j$ iff in $G$ there is a $0$ cost path between $i$ and $j$ having
only non-negative prefix sums.
}
\label{semi-DYCK}
\end{lemma}

The next lemma includes labeled edges going from a node to itself.
This paper assumes no self-cycles with repeated edges.

\begin{lemma}
{\sf
Consider an LDG $G = (\Sigma, V,E_1)$ where $\Sigma$ is Dyck (semi-Dyck), $|\Sigma|= 2$ with sign-closure ${\mbox{\bf unsign}}^{\geq}(G)$  (${\mbox{\bf unsign}}(G)$). 
Then all vertices in $V$ have at most $3n$ outgoing edges.
}
\label{neighborBounds}
\end{lemma}

\section{Towards efficient exact $0$ cost paths}
\label{sec:SP-distances}

This section gives the background for determining which nodes have {\em exact} paths of costs $0$ in LDGs with $\{ \, -1, +1 \, \}$ weighted edges in $\widetilde{O}(n^{\omega})$ operations.
This new solution is expressed as flat Dyck or flat semi-Dyck paths in LDGs.
This is done by computing sign-closures of digraphs with $\{ \, -1, +1 \, \}$ edge weights.
In the process, $0$ cost edges may be added to these diagraphs.
Also, $\pm 1$ and $0$ edges are extended by $0$ cost edges.
Our algorithm uses algebraic matrix multiplication of specially coded matrices.
These matrix encodings are from Alon, Galil, and Margalit~\cite{AGM}. See also Yuval~\cite{Yuval}.
Each algebraic matrix multiplication may be done in $O(n^{\omega} \log n)$. 
This may be improved by a polylog factor, see for example~\cite{AGM,AHU:1974,C,Zwick:2002,Rytter}.

Alon, Galil, and Margalit's shortest path algorithm~\cite{AGM} starts by finding exact $0$ length paths in digraphs with edge costs $\{ -1, 0, +1 \}$.
They use these $0$ exact paths as shortcuts to find shortest paths.
Our algorithm finds $\{-1,0,+1\}$ exact paths between all pairs of vertices.
Alon, et al.'s digraph shortest path algorithm works for edges with much larger costs. 
Their digraph shortest path algorithm is substantially more costly than our digraph exact path algorithms.
Of course, they solve the shortest path algorithm where we solve a reachability problem.

Matrices are written in uppercase and their elements are written in lowercase~\cite{Zwick:2002}.
Matrix parenthesized superscripts, such as those in $\{ D^{(-1)}, D^{(0)}, D^{(+1)} \}$ signify different
matrices. These parenthesized powers are not exponentiation.
Likewise, matrix elements raised to powers, such as $d_{i,j}^{-1}, d_{i,j}^{0}, d_{i,j}^{+1}$, are not exponentiated. 
Rather these superscripts indicate the matrices these elements are from. 
In this case, $d_{i,j}^{-1}$ is in $D^{(-1)}$, $d_{i,j}^{0}$ is in $D^{(0)}$, and $d_{i,j}^{+1}$ is in $D^{(+1)}$.

The algorithm in Figure~\ref{Min-Plus-first} maintains three adjacency matrices $D^{(-1)}, D^{(0)},$ and $D^{(+1)}$.
These three adjacency matrices allow $-1, 0$ and $+1$ edges to go from any vertex to any other vertex.

Given an LDG $G=(\Sigma, V,E_1)$, where $|\Sigma|=2$, define the adjacency matrices $D^{(g)} \in \{ D^{(-1)}, D^{(0)}, D^{(+1)} \}$ whose edge costs
are from $\{ \, -1, 0, +1 \, \}$.
Before iteration $\ell = 1$, there are no $0$ cost edges in $D^{(0)}$.

\begin{eqnarray*}
d_{i,j}^{g}
& = &
\left\{
\begin{array}{rll}
-1 		& \mbox{ if } (i,j) \in E_{\ell} \mbox{ and } g = -1 = \lc(i,j)\\
0		& \mbox{ if } (i,j) \in E_{\ell} \mbox{ and } g = 0 = \lc(i,j) \mbox{ and } \ell > 1\\
+1 		& \mbox{ if } (i,j) \in E_{\ell} \mbox{ and } g = +1 = \lc(i,j)\\
\infty 	& \mbox{ otherwise.}
\end{array}
\right.
%\label{init-Min-Plus}
\end{eqnarray*}
%
%
%\marginpar{Watch $\ell$}
%
%

Subsequently, in each iteration a new edge set $E_{\ell}$ is created in the $\ell$-th iteration of our main algorithm.
At this point, any new $\{ -1, 0, +1 \}$ cost paths are placed in $E_{\ell}$ during iteration~$\ell$.

So, during computation there may be at most three (different) labeled edges directly from any vertex $i$ to any other vertex $j$.
See Lemma~\ref{neighborBounds}.
Recall, the initial graph edges only have weights from $\{ -1, +1 \}$.
The algorithm in Figure~\ref{Min-Plus-first} implements these equations to find all $-1, 0, +1$ exact paths.
This algorithm is substantially less efficient than Nyk\"anen and Ukkonen~\cite{NU} applied to graphs with $\{ -1, +1 \}$ edges.
However, Figure~\ref{Min-Plus-first}'s algorithm forms a basis for our more efficient algorithm.

\begin{figure}[ht]
\begin{center}
  \leavevmode
%% put framebox and vbox on the same line!  Otherwise LaTeX complains!
\framebox[1.05\width][l]{\vbox{
\begin{tabbing}
text \= text \= text \= text \= text \= text \kill
\> {\bf Expensive-Digraph-exact-paths}: $\pm 1$, for the semi-Dyck LDG $G = (V,E_1)$\\
\> 1. $\{ \, D^{(-1)}, D^{(0)}, D^{(+1)} \, \} \leftarrow$ {\bf Init-Adjacency-Matrices}$(G)$\\
\> 2. $n \leftarrow |V|$\\
\> 3. {\bf for} $\ell \leftarrow 2$ {\bf to } $n$ {\bf do}\\ % $\lceil \log n \rceil$
\> 4. \> {\bf for} $i \leftarrow 1$ {\bf to } $n$ {\bf do}\\
\> 5. \> \> {\bf for} $j \leftarrow 1$ {\bf to } $n$ {\bf do}\\
\> 6. \> \> \> {\bf for} $k \leftarrow 1$ {\bf to } $n$ {\bf do}\\
\> 7. \> \> \> \> {\bf if } $(d_{i,k}^{-1}+ d_{k,j}^{+1}) = 0 \vee (d_{i,k}^{+1}+ d_{k,j}^{-1}) = 0$ {\bf then} $d_{i,j}^{0} \leftarrow 0$\\
\> 8. \> \> \> \> {\bf if } $(d_{i,k}^{+1}+ d_{k,j}^{0}) = 1 \vee (d_{i,k}^{0}+ d_{k,j}^{+1}) = 1$ {\bf then} $d_{i,j}^{+1} \leftarrow +1$\\
\> 9. \> \> \> \> {\bf if } $(d_{i,k}^{-1}+ d_{k,j}^{0}) = -1 \vee (d_{i,k}^{0}+ d_{k,j}^{-1}) = -1$ {\bf then} $d_{i,j}^{-1} \leftarrow -1$
\end{tabbing}
}}%% vbox  %%framebox
\end{center}
\vspace{-0.125in}
\caption{An inefficient $-1, 0, +1$ exact path algorithm for digraphs with initial edge costs $\{ \, -1, +1 \, \}$}
\label{Min-Plus-first}
\end{figure}

The function {\bf Init-Adjacency-Matrices} in Figure~\ref{Min-Plus-first} initializes each of $D^{(-1)}, D^{(0)},$ and $D^{(+1)}$ with sufficiently large
values representing no edge and no path.
No path and no edge in the algorithm in Figure~\ref{Min-Plus-first} may be represented by numbers as little as~$2$.
Although, if there is no $t$ to get from $i$ to $j$, then $d^{t}_{i,j}$ is effectively infinite, for any $t \in \{ -1,0,+1 \}$.
Next {\bf Init-Adjacency-Matrices}
represents a $-1$ edge from $i$ to $j$ in $D^{(-1)}$ by placing $-1$ in $d^{-1}_{i,j}$.
Likewise $+1$ edges are represented in $D^{(+1)}$ by appropriate placement of $+1$ values.

\begin{definition}[$E_1$-length]
  {\sf
  Consider an LDG $G = (V,E_1)$ and an exact path $p$ in $G$.
    The {\em $E_1$-length} of $p$ is the number of $E_1$ edges in $p$.
  }
  \label{E1-length}
\end{definition}

Exact $0$ cost paths have even $E_1$-lengths.
Exact $\pm 1$ cost paths have odd $E_1$-lengths.

The algorithm in Figure~\ref{Min-Plus-first} may be made a little more efficient.
In the next section of this paper, we give a substantially more efficient solution building on this approach.

\begin{lemma}
{\sf
Consider an LDG $G = (\Sigma, V, E_1)$, where $\Sigma$ is semi-Dyck with $|\Sigma| = 2$.
Then at the termination of Figure~\ref{Min-Plus-first}'s algorithm,  $d_{i,j}^{w} = w$, for $w \in \{ -1, 0, +1 \}$, for all $\{ i,j \} \subseteq V$
where there is an exact $w$ cost path from $i$ to $j$.
}
\label{Min-Plus-First-Lemma}
\end{lemma}

\begin{proof}
This is shown by complete induction on the iteration $\ell$ for exact~$-1, 0$ and $+1$ paths.\\

\noindent
{\bf Basis} 
Immediately after the initial iteration $\ell =2$, all $0$ exact paths of even $E_1$-length at least $2$ are found by line~7.
Such exact $0$ paths are created by combining adjoining $+1$ and $-1$ edges or combining adjoining $-1$ and $+1$ edges.
An exact $0$ path from $i$ to $j$ is recorded in the matrix $D^{(0)}$ by setting $d^{0}_{i,j}$ to $0$.

In iteration $\ell =3$ the exact $0$ cost paths from iteration $\ell =2$ are combined with adjoining $\pm 1$ edges from $E_1$ giving $\pm 1$ exact paths.
These exact $\pm 1$ paths have odd $E_1$-length of at least~3.
This is done by lines~8 and~9 and the matrices $D^{(\pm 1)}$ record these exact paths.\\

\noindent
{\bf Inductive Hypothesis} For $\ell =2$ and $\ell = 3$, the next cases hold.

 Immediately after iteration $\ell: \ell \geq t \geq 1$, for all even $t$, this algorithm finds all exact $0$
cost paths in line~7. 
By assumption, for all even $t$, these new exact $0$ paths discovered in iteration $t$ have even $E_1$-length of at least $t$.
Line~7 combines adjoining $\pm 1$ and $\mp 1$ exact paths from previous iterations and records these paths in $D^{(0)}$.
These new exact $0$ paths are recorded in $D^{(0)}$.

After iteration $\ell: \ell \geq t \geq 1$, for all odd $t$, this algorithm finds all exact $\pm 1$
cost paths in lines~8 and~9. 
By assumption, for all odd $t$, these new exact $\pm 1$ paths discovered in iteration $t$ have odd $E_1$-length of at least $t$.
Lines~8 and~9 combines exact $0$ paths and exact $\pm 1$ paths from previous iterations.
These new exact $\pm 1$ paths are recorded in $D^{(\pm 1)}$.\\

\noindent
{\bf Inductive Step} Consider the algorithm immediately after iteration $\ell+1$ where $\ell +1$ is even.
By the inductive hypothesis, consider all odd $t: \ell \geq t \geq 1$, the matrices $D^{(\pm 1)}$ contain $\pm 1$ 
exact paths of odd $E_1$-length.
Also, all $\mp 1$ exact paths in $D^{(\mp 1)}$ are of odd $E_1$-length.
In iteration $\ell +1$, this algorithm combines adjoining $\pm 1$ and $\mp 1$ exact paths to form $0$ exact
paths of even $E_1$-length.
Suppose an exact $0$ cost path $p$ is discovered in iteration $\ell+1$ where $p$ is of $E_1$-length $t \leq \ell-1$.
This cannot be the case since by the inductive hypothesis $p$ would have been discovered in iteration $t$.

Consider the algorithm immediately after iteration $\ell+1$ where $\ell +1$ is odd.
By the inductive hypothesis for all even $t: \ell \geq t \geq 1$, the matrices $D^{(0)}$ contain 
exact $0$ paths of even $E_1$-length.
Likewise, for odd $t: \ell \geq t \geq 1$, the matrices $D^{(\pm 1)}$ contain $\pm 1$ exact $\pm 1$ paths of 
odd $E_1$-length.
In this case, the algorithm combines adjoining $\pm 1$ ($0$) exact paths with $0$ ($\pm 1$) exact paths giving new 
exact $\pm 1$ cost paths of odd $E_1$-length.
Suppose an exact $\pm 1$ cost path $p$ is discovered in iteration $\ell+1$ where $p$ is of $E_1$-length $t \leq \ell-1$.
This cannot be the case since by the inductive hypothesis $p$ would have been discovered by iteration $t$.
\end{proof}

\begin{lemma}
{\sf
Consider an LDG $G = (\Sigma, V, E_1)$, where $\Sigma$ is semi-Dyck with $|\Sigma| = 2$, and the algorithm in Figure~\ref{Min-Plus-first}.
At the termination of the algorithm, if $d_{i,j}^{w} = w \in \{ -1, 0, +1 \}$, 
for any $\{ i,j \} \subseteq V$ where there is a $w$ cost exact path from $i$ to $j$, then
the algorithm computed the sign-closure $\mbox{\bf unsign}(G)$.
}
\label{Special-Prod-Lemma}
\end{lemma}

\begin{proof}
At the termination of the algorithm $d_{i,j}^{w} = w$, where $w \in \{ -1, 0, +1 \}$, for all $\{ i,j \} \subseteq V$ where there is an exact $w$ cost path from $i$ to $j$.
If $\mbox{\bf unsign}(G)$ is not complete, then some edge from 
$i$ to $j$ must not have been placed in $E({\bf unsign}(G))$, by applying the sign-closure rule

\begin{tabbing}
text\= text \= text \= text \= text \= text \kill
\> {\bf if} $i \stackrel{v}{\longrightarrow} k \stackrel{w}{\longrightarrow} j \in E({\bf unsign}(G))$ and ${\bf sgn}(v) \neq {\bf sgn}(w)$\\
\> {\bf then} put $i \stackrel{v+w}{\longrightarrow} j$ in $E({\bf unsign}(G))$.
\end{tabbing}

Since $v \neq w$, 
then there must be some path from $i$ to $j$ that was not generated by the algorithm
in Figure~\ref{Min-Plus-first}.

But, $i \stackrel{v+w}{\longrightarrow} j$ is an exact $v+w \in \{ -1, 0, +1 \}$ path from $i$ to $j$.
This exact path must have been found by Lemma~\ref{Min-Plus-First-Lemma}, completing the proof.
\end{proof}

Nyk\"anen and Ukkonen's sign-closure algorithm~\cite{NU} finds exact paths in graphs in $O(n^3)$ time.
The next result shows how to find Dyck exact $-1, 0,$ and $+1$ paths.
This is done by dropping the condition $d_{i,k}^{-1}+ d_{k,j}^{+1}$ from the calculation of $d_{i,j}^{0}$
in Figure~\ref{mainEQs}.

\begin{lemma}
{\sf
Consider an LDG $G = (\Sigma, V, E_1)$, where $\Sigma$ is Dyck with $|\Sigma| = 2$, and the algorithm in Figure~\ref{Min-Plus-first}
while dropping the expression $(d_{i,k}^{-1}+ d_{k,j}^{+1}) = 0$ in line~7.
At the termination of the algorithm, if $d_{i,j}^{w} = w \in \{ -1, 0, +1 \}$, for any $\{ i,j \} \subseteq V$
where there is a $w$ cost Dyck path from $i$ to $j$, then the algorithm 
computed the Dyck sign-closure $\mbox{\bf unsign}^{\geq}(G)$.
}
\label{inefficient-Dyck}
\end{lemma}

\begin{proof}
Lemma~\ref{Min-Plus-First-Lemma} shows this algorithm finds all $-1,0,$ and $+1$ exact paths for any semi-Dyck LDG $G$.
With this in mind, it remains to extend that lemma.
The next arguments allows the extension of Lemma~\ref{Min-Plus-First-Lemma}'s induction proof to this Dyck case.

Lines 7, 8, and~9 in Figure~\ref{Min-Plus-first}, only computing $d_{i,j}^{0}$ may create a negative prefix sum for a $0$ cost path or subpath.
Clearly computing $d_{i,j}^{+1}$ cannot have a negative prefix sum.
Likewise, computing $d_{i,j}^{-1}$ can't compute a negative prefix sum for a $0$ cost path or subpath.
Computing $d_{i,j}^{-1}$ just extends $-1$ edges, but does not necessarily contribute to non-Dyck labeled paths.

Removing the condition $(d_{i,k}^{-1}+ d_{k,j}^{+1}) = 0$ in line~7 in the equation for $d_{i,j}^{w}$ finds 
all exact $w \in \{ -1, 0, +1 \}$ cost paths without negative prefix sums. 
An identical argument as in the proof of Lemma~\ref{Special-Prod-Lemma} indicates
all edges of $\mbox{\bf unsign}^{\geq}(G)$ have been found.
Thus, this modification computes the Dyck sign-closure.
\end{proof}

Intuitively, our approach to improving the algorithm in Figure~\ref{Min-Plus-first} is anchored in Boolean matrix multiplication for transitive closure.
However, starting with graph edges $\{ -1, +1 \}$ and computing with the edge weights $\{ -1, 0, +1 \}$ seems to 
preclude Boolean matrix multiplication.
Thus, we leverage Alon, Galil, and Margalit~\cite{AGM}.

\subsection{AGMY matrix encoding}
\label{AGMY-encoding}
Alon, Galil, and Margalit~\cite{AGM} as well as Yuval~\cite{Yuval} supply the basis of our {\em (AGMY)} algebraic matrix coding.
These AGMY style codings have been very fruitful, see for example~\cite{Zwick:2002, SZ, ROMANI1980134, Takaoka1998,GM}.

Lemma~\ref{neighborBounds} gives insight into an algebraic matrix product solution.
In particular, the AGMY representation uses powers of $3(n+1)$ to differentiate $\{ -1, 0, +1 \}$ edge weights.
That is, $\frac{1}{3(n+1)}, (3(n+1))^{0}, 3(n+1)$ represent $-1, 0, +1$ edges, respectively.
These AGMY values are sufficiently separated to allow information to be gleaned after an algebraic matrix product.

Figure~\ref{coding} shows how to translate adjacency matrices $D^{(-1)}, D^{(0)}$ and $D^{(+1)}$ to an AGMY encoded
adjacency matrix.
The restriction $g \neq h$ is from the sign-closure in Definition~\ref{Def:Closure}.

\begin{figure}[ht]
\begin{center}
  \leavevmode
\framebox[1.00\width][l]{\vbox{
\begin{eqnarray*}
c_{i,j} & \leftarrow & 
\left\{
\begin{array}{ll}
\displaystyle \sum_{\stackrel{k=1}{g \neq h}}^{n} (3(n+1))^{d_{i,k}^{g}+d_{k,j}^{h}} & \mbox{ if } d_{i,k}^{g} \neq \infty \wedge d_{k,j}^{h} \neq \infty\\[1cm]
0 		                                                 & \mbox{ if }  d_{i,k}^{g} = \infty \vee d_{k,j}^{h} = \infty.\\
\end{array}
\right.
\end{eqnarray*}
}
}
\end{center}
\caption{AGMY matrix coding for algebraic matrix multiplication to simuate one matrix dot product based on Alon, Galil, and Margalit~\cite{AGM}; and Yuval~\cite{Yuval}}
\label{coding}
\end{figure}

An algebraic matrix product computes the expression in Figure~\ref{coding}
for all $i,j: n \geq i,j \geq 1$, see Figure~\ref{mainEQs}.
If $d_{i,k}^{g} = \infty$ or $d_{k,j}^{h} = \infty$, then replace $(3(n+1))^{d_{i,k}^{g}+d_{k,j}^{h}}$ with $0$.
This works since no finite power of $3(n+1)$ is~$0$.

\begin{figure}[ht]
\framebox[1.00\width][l]{\vbox{
\vspace{0.1in}
For all $i,j: n \geq i,j \geq 1$, let
\begin{eqnarray*}
c_{i,j}^{-1} & \leftarrow &  (3(n+1))^{d_{i,k}^{0}+d_{k,j}^{-1}} + (3(n+1))^{d_{i,k}^{-1}+d_{k,j}^{0}}\\
c_{i,j}^{0} & \leftarrow &  (3(n+1))^{d_{i,k}^{-1}+d_{k,j}^{+1}} + (3(n+1))^{d_{i,k}^{+1}+d_{k,j}^{-1}}\\
c_{i,j}^{+1} & \leftarrow &  (3(n+1))^{d_{i,k}^{0}+d_{k,j}^{+1}} + (3(n+1))^{d_{i,k}^{+1}+d_{k,j}^{0}}
\end{eqnarray*}
}
}
\caption{A breakout of computing AGMY matrix values}
\label{mainEQs}
\end{figure}

Initially, an adjacency matrix represents an LDG $G$ with edge costs from $\{-1, +1\}$.
So at the start of the algorithm, any two vertices $i$ and $j$ may share a $-1$ and $+1$ edge going in each direction.
Thus, each element of the initial AGMY coded adjacency matrices starts with values from,
\begin{eqnarray*}
\{ \, 0,  \, (3(n+1))^{1},  \, (3(n+1))^{-1}, \, (3(n+1))^{1} + (3(n+1))^{-1} \, \}.
\end{eqnarray*}
Here the AGMY $0$ represents no edge.

During the first matrix product, $0$ cost edges may appear.
They are represented by $(3(n+1))^0 = 1$.

The ideas for the next Lemma are based on Alon, Galil, and Margalit~\cite{AGM}.

\begin{lemma}
  {\sf
    Given two $n \times n$ AGMY encoded LDG adjacency matrices $S$ and $T$ representing values from $\{-1,0,+1\}$.
    Consider an algebraic matrix product $P = S \, T$ and say there is a path of cost $p_{i,j}$ between $i$ and $j$, then
    \begin{eqnarray*}
    n \left[ (3(n+1))^2 + 2(3(n+1)) + 3 + \frac{2}{3(n+1)} + \frac{1}{(3(n+1))^2} \right] & \geq &  p_{i,j}.
    \end{eqnarray*}
}
\label{mm-bound}
\end{lemma}

\begin{proof}
While a sign-closure is computed, any two vertices $i$ and $j$ may share up to three edges going in each direction by Lemma~\ref{neighborBounds}.
In AGMY coding, three outgoing edges are bounded by,
\begin{eqnarray*}
B & \leq & \frac{1}{3(n+1)} + 1 + 3(n+1).
\end{eqnarray*}

Thus, a single algebraic matrix product produces a new matrix element of at most
\begin{eqnarray}
(3(n+1))^2 + 2(3(n+1)) + 3 + \frac{2}{3(n+1)} + \frac{1}{(3(n+1))^2} & \geq & B^2.
\label{UBound}
\end{eqnarray}

The dot-product of row $S[i,*]$ and column $T[*,j]$ gives a value of at most,
\begin{eqnarray*}
\underbrace{B^2 + \cdots + B^2}_{\mbox{the sum of $n$ squares}}
\end{eqnarray*}
and since there are at most $n$ of these $B^2$ terms, the result holds.
\end{proof}
\vspace{0.25in}

Following the AGMY adjacency matrices of Lemma~\ref{mm-bound},
Say there is a path from $i$ to $j$, then the dot-product of row $S[i,*]$ and column $T[*,j]$ is at least AGMY $\frac{1}{3(n+1)}$.
This is the result of combining adjoining $-1$ and $0$ edges.
This is because a $-1$ cost edge is represented by AGMY $\frac{1}{3(n+1)}$
and a $0$ edge is represented by AGMY $(3(n+1))^0 = 1$.

Factors besides $\{ \frac{1}{3(n+1)}, 1, 3(n+1) \}$ in Lemma~\ref{mm-bound} are removed after each matrix product.
These factors represent unnecessary intermediary paths and their growth makes the algorithm too expensive.
So, following Alon, et al.~\cite{AGM}, see also Zwick~\cite{Zwick:2002}, our algorithm removes all edges except
for AGMY $\{ \frac{1}{3(n+1)}, 1, 3(n+1) \}$.
This is normalization. Normalization removes unnecessary edges for computing the EPL for $0$ cost paths.
So, immediately after the matrix product in line~5, the adjacency elements are converted back to values from $\{ \frac{1}{3(n+1)}, 1, 3(n+1) \}$.

The next corollary follows from the upper bound on the representation of each adjacency element from Lemma~\ref{mm-bound}.
Very similar results are in~\cite{AGM,Zwick:2002}.

\begin{corollary}
{\sf
In a single AGMY algebraic matrix multiplication of an LDG's adjacency matrix,
each of the resulting matrix elements may be represented in $O(\log n)$ bits.
}
\end{corollary}

New $\pm 1$ edges are generated by matrix products and normalization or extending $\pm 1$ by $0$ cost
edges.
Furthermore, in computing Dyck paths, any edge $i \stackrel{-1}{\longrightarrow} k$ edge may {\em not} 
be joined to $k \stackrel{+1}{\longrightarrow} j$. 
This is because 

$$i \stackrel{-1}{\longrightarrow} k \stackrel{+1}{\longrightarrow} j$$ 

\noindent
is not Dyck.

In the case when a $-1$ edge is made from multiple $E_1$ edges, then this edge represents a path. 
Therefore, this path has a negative sum. In fact, it is a negative prefix sum. 
Such a $-1$ edge is already not Dyck.
Thus, it may not start a new Dyck path, though it may follow a $+1$ edge.

\begin{figure}[ht]
\begin{center}
  \leavevmode
%% put framebox and vbox on the same line!  Otherwise LaTeX complains!
\framebox[1.05\width][l]{\vbox{
\begin{tabbing}
text \= text \= text \= text \= text \= text \kill
\> 1. {\em detectNegativeOneEdge}$(\mbox{edge\_cost}, n)$\\
\> 2. \> $\mbox{check} \leftarrow 3(n+1) \times \mbox{fractional\_part}(\mbox{edge\_cost})$\\
\> 3. \> {\bf if } $2n \geq \mbox{check} \geq 1$ {\bf then return} True\\
\> 4. \> {\bf else return} False\\
\> \\
\> 1. {\em detectPositiveOneEdge}$(\mbox{edge\_cost}, n)$\\
\> 2. \> $\mbox{check} \leftarrow \mbox{truncate}(\mbox{edge\_cost}/ 3(n+1))$\\
\> 3. \> {\bf if } $2n \geq \mbox{check} \geq 1$ {\bf then return} True\\
\> 4. \> {\bf else return} False\\
\> \\
\> 1. {\em detectZeroEdge}$(\mbox{edge\_cost}, n)$\\
\> 2. \> $\mbox{check} \leftarrow \mbox{truncate}(\mbox{edge\_cost}) \, {\bf mod} \, 3(n+1)$\\
\> 3. \> {\bf if } $3n \geq \mbox{check} > 0$ {\bf then return} True\\
\> 4. \> {\bf else return} False
\end{tabbing}
}}%% vbox  %%framebox
\end{center}
\vspace{-0.125in}
\caption{Functions used to detect AGMY costs representing $\{ -1, 0, +1 \}$ for normalize an AGMY algebraic matrix product}
\label{afterRD}
\end{figure}

Normalization uses the functions in Figure~\ref{afterRD}.
The upper and lower bounds in each function in Figure~\ref{afterRD} are determined as follows.
Lemma~\ref{mm-bound} gives upper bounds for $\mbox{\em detectNegativeOneEdge}(\mbox{edge\_cost}, n)$.
That is, the first operation of $\mbox{\em detectNegativeOneEdge}$ is to multiply the fractional part of its edge-weight by $3(n+1)$ giving the upper-bound
\begin{eqnarray*}
2n  & = & 3(n+1) \left(\frac{2n}{3(n+1)} \right).
\end{eqnarray*}

An AGMY $\frac{2}{3(n+1)}$ term in Equation~\ref{UBound} indicates there are at most two different ways to form a $-1$ edge edge through a single intermediary vertex.
For example, take paths from a vertex~$i$ to another vertex~$j$ with intermediary $k$.
That is, the two ways are: $i \stackrel{-1}{\longrightarrow} k \stackrel{0}{\longrightarrow} j$ or $i \stackrel{0}{\longrightarrow} k \stackrel{-1}{\longrightarrow} j$.
Equation~\ref{UBound} also indicates there is at most one cost $-2$ edge between any two vertices due to the $\frac{1}{(3(n+1))^2}$ term.
A potential $-2$ cost edge will not be detected by $\mbox{\em detectNegativeOneEdge}$
because a single $-1$ edge that has AGMY cost of at least $\frac{1}{3(n+1)}$.
Thus, the first operation of $\mbox{\em detectNegativeOneEdge}$ is to multiply the fractional part of the AGMY edge-weight by $3(n+1)$
so the smallest value of a $-1$ cost edge is $1$.
The upper and lower bounds in $\mbox{\em detectPositiveOneEdge}$ and $\mbox{\em detectZeroEdge}$ are similar.

Thus, reset all elements immediately after each recursive doubling step using the functions in Figure~\ref{afterRD}.
The three functions in Figure~\ref{afterRD} detect $-1,0$ and $+1$ cost edges
following each algebraic AGMY matrix multiplication.
These AGMY edges have values as large and complex as those in Lemma~\ref{mm-bound}.
After detecting $-1, 0$ or $+1$ AGMY edges, more complex pathways are simplified or normalized by replacing them by the appropriate members of $\{ \frac{1}{3(n+1)}, 1, 3(n+1) \}$.
Each normalization costs~$O(n^2 \log n)$.

\begin{figure}[!t]
\begin{center}
  \leavevmode
%% put framebox and vbox on the same line!  Otherwise LaTeX complains!                                                                                                
\framebox[1.05\width][l]{\vbox{
\begin{tabbing}
a \= text \= text \= text \= text \= text \kill
\> {\bf Digraph-flat-exact-paths}($G$)\\
\> // $G = (\Sigma, V,E_1)$ is an LDG, $\Sigma$ is semi-Dyck and $|\Sigma| = 2$\\
\> 1. $\{ D^{(-1)}, D^{(0)}, D^{(+1)} \} \leftarrow$ Init-Adjacency-Matrices$(G)$\\
\> 2. $M \leftarrow \mbox{AGMY-Code-then-Sum}(D^{(-1)}, D^{(0)}, D^{(+1)})$\\
\> 3. $n \leftarrow |V|$\\
\> 4. {\bf for} $\ell \leftarrow 2$ {\bf to } $\lceil \log n \rceil + 1$ {\bf do}\\
\> 5. \> $M \leftarrow M \, M$ \ // AGMY, find new $0$ edges\\
\> 6. \> Remove $\pm 1$ edges from $M$\\
\> 7. \> $M \leftarrow \mbox{\em Normalize\_and\_Divide\_by\_2}(M)$\\
\> 8. \> $Z \leftarrow \mbox{Get-Zero-Edges}(M)$\\
\> 9. \> $M \leftarrow Z \, M \, Z$ \ // AGMY, extend $\pm 1$ and $0$ edges
\end{tabbing}
}}%% vbox  %%framebox                                                                                                                                                 
\end{center}
\vspace{-0.125in}
\caption{A new matrix-based sign-closure algorithm for digraphs with initial edge costs representing $\{ \, -1, +1 \, \}$. This algorithm
uses AGMY coded matrices.}
\label{New-Min-Plus}
\end{figure}

Figure~\ref{New-Min-Plus} is the critical component of all our results.
In Figure~\ref{New-Min-Plus}, 
$\mbox{\em Normalize\_and\_Divide\_by\_2}$ removes redundant edges.
By Lemma~\ref{mm-bound}, in line~5 the algebraic AGMY matrix multiplication gives values as large as
\begin{eqnarray*}
n\left[ (3(n+1))^2 + 2(3(n+1)) + 3 + \frac{2}{3(n+1)} + \frac{1}{(3(n+1))^2} \right].
\end{eqnarray*}
Line~6 removes $\pm 1$ edges. In line~7, normalization changes $\pm 2$ cost edges to $\pm 1$ cost edges.
Thus, only retaining AGMY encodings for $\{ \frac{1}{3(n+1)}, 1, 3(n+1) \}$.
The idea of dividing the edge costs by~$2$ is from Alon, et al.~\cite{AGM}.

Line~9 joins adjacent $0$ cost edges and it extends $\pm 1$ cost edges with adjoining $0$ cost edges.
Line~9 cannot generate $\pm 2$ cost edges.
Before line~9 in iteration~$\ell$, these adjoining $0$ cost edges have
$E_1$-length from $2$ up to $3^{\ell -2} \cdot 2^{\ell -1}$.
So, at the end of line~9 in iteration~$\ell$, up to three consecutive adjoining $0$ cost edges may form a single $E_1$-length $3^{\ell-1} \cdot 2^{\ell-1}$ edge.
%% In particular, as many as three $0$ cost adjoining edges may be joined by line~9 giving 
%% \begin{eqnarray*}
%% 3^{\ell -2} \cdot 2^{\ell -1} + 3^{\ell -2} \cdot 2^{\ell -1} + 3^{\ell -2} \cdot 2^{\ell -1}
%% & = & 3 \cdot 3^{\ell -2} \cdot 2^{\ell -1}\\
%% & = & 3^{\ell-1} \cdot 2^{\ell-1}.
%% \end{eqnarray*}

\begin{figure}[ht]
\begin{center}
  \leavevmode
%% put framebox and vbox on the same line!  Otherwise LaTeX complains!
\framebox[1.05\width][l]{\vbox{
\begin{tabbing}
a \= text \= text \= text \= text \= text \kill
\> 5'. \> $M' \leftarrow \mbox{\em Markup\_minus\_one\_edges}(M)$ \ // AGMY, mark $-1$ edges\\
\> 5''. \> $M \leftarrow M' \, M$ \ // AGMY product, non-Dyck $0$ edges are detectable
\end{tabbing}
}
}
\end{center}
\vspace{-0.125in}
\caption{Replacing line 5 of Algorithm~\ref{New-Min-Plus} with these lines is the basis of a sign-closure algorithm for Dyck digraphs with initial edge costs $\{ \, -1, +1 \, \}$. 
Further updates are done in $\mbox{\em Normalize\_and\_Divide\_by\_2}(M)$ so it deletes $0$ edges marked here as created by non-Dyck paths.}
\label{New-Min-Plus-Dyck}
\end{figure}

Figure~\ref{New-Min-Plus-Dyck} shows how to determine if a new $0$ edge is made by a $-1$ edge followed by a $+1$ edge.
A $-1$ edge followed by a $+1$ edge is not Dyck.
The function $\mbox{\em Markup\_minus\_one\_edges}$ adds $\left( \frac{1}{3(n+1)} \right)^4$ to each negative edge in its AGMY input matrix $M$ producing $M'$.
Now, a $-1$ cost edge followed by a $+1$ cost edge has an AGMY product of 
\begin{eqnarray*}
\left( \left( \frac{1}{3(n+1)} \right)^4 + \left( \frac{1}{3(n+1)} \right) \right) 3(n+1) & = &  \left( \frac{1}{3(n+1)} \right)^3 + 1
\end{eqnarray*}
and such cubic terms are found by a Dyck modified $\mbox{\em Normalize\_and\_Divide\_by\_2}(M)$ in the algorithm in Figure~\ref{New-Min-Plus}.
This cubic term only occurs when an AGMY $-1$ value is multiplied by an AGMY $+1$ value, when this product represents a $-1$ edge going to a $+1$ edge.
Of course, a $-1$ edge going to a $+1$ edge is not a Dyck path.
Also these cubic terms cannot become $-2$ AGMY edges in one AGMY matrix multiplication, by Lemma~\ref{mm-bound}.
Finally, for the Dyck case, the function $\mbox{\em Normalize\_and\_Divide\_by\_2}(M)$ in the algorithm in Figure~\ref{New-Min-Plus-Dyck} is updated to delete any $0$ 
cost edge made from a $-1$ edge followed by a $+1$ edge.

The next convention smooths the subsequent presentation.

\begin{convention}[{\bf Digraph-flat-exact-paths} context]
{\sf
The algorithm {\bf Digraph-flat-exact-paths} refers to Figure~\ref{New-Min-Plus} for the semi-Dyck case in addition to 
Figure~\ref{New-Min-Plus-Dyck} for the Dyck case, depending on the context.
}
\end{convention}

\subsection{Flat Dyck and semi-Dyck grammars}
\label{flat-grammars}

Flat grammars supply a foundation for the complete solution.

\begin{definition}[Flat Dyck and flat semi-Dyck grammars]
{\sf
The flat Dyck language is,
\begin{eqnarray*}
F & \Longrightarrow & T \ | \  P \\
T & \Longrightarrow & T \, T \ | \  P \ | \ \epsilon\\
P & \Longrightarrow & {\bf a} \, P \, {\bf a}^{-1}  \ | \ \epsilon.
\end{eqnarray*}
The flat semi-Dyck language replaces the last production with $P \ \Longrightarrow \ {\bf a} \, P \, {\bf a}^{-1} \ | \ {\bf a}^{-1} \, P \, {\bf a}  \ | \ \epsilon$.
}
\label{flatGrammars}
\end{definition}

\begin{definition}[Sign-closure graph evolution]
  {\sf
    Consider an LDG $G_1 = (\Sigma, V, E_1)$ and {\bf Digraph-flat-exact-paths}. % the algorithm in Figure~\ref{New-Min-Plus}.
    At the end of each iteration $\ell \geq 2$ this algorithm produces the LDG $G_{\ell} = (\Sigma \cup \{ 0 \}, V, E_{\ell})$.
  }
  \label{G-ell-graphs}
\end{definition}

A path in $G$ may become a single $\{ -1,0,+1 \}$ edge in $E_{\ell}$.
Such new edges are generated by {\bf Digraph-flat-exact-paths}.
The number of $+1$ and $-1$ edges from $E_1$ contributing to new edges give important insights.

\begin{definition}[$c_{+}$-length and $c_{-}$-length]
  {\sf
  Consider an LDG $G_{\ell} = (\Sigma, V, E_{\ell})$ and an edge $e \in E_{\ell}$, for $\ell \geq 1$.
  An edge $e$'s $c_{+}(e)$-length is  $e$'s number of $+1$ edges from $E_1$ and $c_{-}(e)$-length is $e$'s number of $-1$ edges from $E_1$.
  }
  \label{C-length}
\end{definition}

In iteration $\ell$ of {\bf Digraph-flat-exact-paths}, edges with $\pm 1$ costs are made by joining $\pm 1$ edges from
iteration $\ell-1$.
Also, additional exact $0$ cost paths may be included in the new edges.
After iteration $\ell =2$, new $\pm 1$ edges are not $\pm 1$ exact paths.

The next corollary is known in a number of contexts.

\begin{corollary}
{\sf
Consider an LDG $G_1 = (\Sigma, V, E_1)$ where $|\Sigma|= 2$, $\Sigma$ is Dyck (semi-Dyck) and
any $0$ cost edge $e$ in $G_{\ell}$ with $E_1$-length $|e|$, then $e$ has $\frac{|e|}{2}$ edges with $+1$ costs
and $\frac{|e|}{2}$ edges with $-1$ costs, all from $E_1$.
}
\label{balance}
\end{corollary}

Corollary~\ref{balance} indicates $c_{+}(e) = c_{-}(e) = \frac{|e|}{2}$ for all $0$ cost edges $e$.
A $c_{+}$-length ($c_{-}$-length) may be any integer from $0$ to $n-1$.

% Let $e_{+} = s \xrightarrow{+1} t$ (or $e_{-} = s \xrightarrow{-1} t$) be in $E_{\ell-1}$ at the end of iteration $\ell-1$.
% The edge $e_{+}$ ($e_{-}$) is {\em incorporated} in all adjoining edges in $E_{\ell}$ when all $E_{\ell -1}$ edges going to $s$ and all $E_{\ell -1}$ edges
% coming from $t$ form new edges by joining with $e_{+}$ ($e_{-}$) in $E_{\ell}$.

% So, $e_{+}$ combines with adjoining edges of the opposite sign giving $0$ edge(s) in iteration $\ell$.
% Additionally, $e_{+}$ combines with adjoining edges with the same sign giving AGMY $+ 2$ edges. 
% Normalization converts $+2$ AGMY edges into $+1$ edges in iteration $\ell$.

The proof of the next lemma uses the idea that a new $\pm 1$ edge $e$ may merge with
an exact $0$ cost edge producing a new $\pm 1$ edge $e'$ with a larger $E_1$-length.
However both edges $e$ and $e'$ have the same $c_{\pm}$-length.

\begin{lemma}
{\sf
  Given an LDG $G_1 = (\Sigma, V, E_1)$, where $\Sigma$ is flat Dyck (flat semi-Dyck) and $|\Sigma| = 2$,
  then all exact $0$ cost edges created by line~5 in iteration $\ell$ of {\bf Digraph-flat-exact-paths} have
  $E_1$-length at least $2^{\ell-1}$, where $\lceil \log n \rceil +1 \geq \ell \geq 2$.
  }
  \label{powers-of-two}
\end{lemma}

\begin{proof}
	Without loss, this proof focuses on the algorithm in Figure~\ref{New-Min-Plus}.
	Line~9 extends $\pm 1$ and exact $0$ cost edges using $0$ edges of $E_1$-length
	from $2$ to $3^{\ell -2} \cdot 2^{\ell -1}$. 
        It extends adjoining pairs of exact $0$ cost paths in flat grammars.
        Thus, line~9 is not in the next induction.
	The induction is on the iteration $\ell$ and includes $\{ -1, 0, +1 \}$ edges of $E_1$-length of at least $2^{\ell-1}$.\\

    \noindent
        {\bf Basis} In iteration $\ell = 2$, line~5 
		computes exact $0$ cost edges of $E_1$-length at least $2^{\ell -1} = 2$. 
		Likewise, line~5 generates all $\pm 2$ cost edges with $E_1$-length of at least $2^{\ell -1} = 2$.
		These $\pm 2$ cost edges are converted to $\pm 1$ cost edges, by normalization in line~7.
        Their $E_1$-lengths remain at least $2^{\ell-1}$, but their $c_{+}$-length or $c_{-}$-length is $2^{\ell -1}$.\\
		
    \noindent
        {\bf Inductive Hypothesis} Assume for some $\lambda$, all iterations $\ell$ where $\lambda \geq \ell \geq 2$ 
		 are such that line~5 computes $\{ -1, 0, +1 \}$ cost edges of $E_1$-length at least $2^{\ell -1}$.
		 Here the $0$ cost edges are exact $0$ cost paths.
                 The new $\pm 1$ edges have $c_{+}$-length and $c_{-}$-length of $2^{\ell -1}$, respectively.\\

    \noindent
        {\bf Inductive Step}  Consider iteration $\lambda +1$ for some $\lambda$ where $\lambda \geq \ell \geq 2$.

		By the inductive hypothesis, in iteration $\lambda = \ell$, line~5 
		computes exact $0$ cost edges of $E_1$-length at least $2^{\lambda-1}$.
		
		In iteration $\lambda = \ell$, $\mbox{\em Normalize\_and\_Divide\_by\_2}$ produces new $\pm 1$ edges only if they 
		are $\pm 2$ edges just generated by line~5.
		These new $\pm 1$ edges have $E_1$ length of at least $2^{\lambda-1}$
		by the inductive Hypothesis.
                Also the inductive Hypothesis indicates these new $\pm 1$ edges have
                $c_{+}$-length or $c_{-}$-length of $2^{\ell -1}$, respectively.
                		
		In conclusion, during iteration $\lambda +1$, line~5 combines adjoining $+1$ and $-1$ cost edges
		forming exact $0$ cost edges with $E_1$-length at least $2 \cdot 2^{\lambda-1} = 2^{\lambda}$.
                Also, the new $+1$ edges have $c_{+}$-length $2^{\lambda}$, and the new $-1$ edges have $c_{-}$-length $2^{\lambda}$.
                This is because new $\pm 2$ edges are made by joining $\pm 1$ edges from the previous iteration.
\end{proof}

The algorithm {\bf Digraph-flat-exact-paths} uses $O(\log n)$ algebraic matrix multiplications. 
The same is true for the Dyck extension given in Figure~\ref{New-Min-Plus-Dyck}.
Each matrix multiplication costs $O(n^{\omega} \log n)$. This gives a total cost of
$O(n^{\omega} \log^2 n)$ time for the flat Dyck (semi-Dyck) case.

\begin{theorem}
{\sf
  Given an LDG $G = (\Sigma, V, E_1)$, where $\Sigma$ is flat Dyck (flat semi-Dyck) and $|\Sigma| = 2$,
  then Figure~\ref{New-Min-Plus}'s algorithm (updated by Figure~\ref{New-Min-Plus-Dyck}) finds all
  exact $0$ paths in $\widetilde{O}(n^{\omega})$ time.
}
\end{theorem}

\section{Exact $0$ cost paths and Dyck and semi-Dyck grid graphs}
\label{sec:grid-graphs}

Dyck and semi-Dyck languages are generalizations of the Flat Dyck and flat semi-Dyck languages.
Grid paths enable the transition from flat grammars to the general case. 
Each acyclic path in an LDG has an equivalent path in a grid graph.
See an example grid graph in Figure~\ref{gridGraph}.

Pyramids and valleys are distinct grid graphs.
Pyramid paths are generated by the non-terminal $P$ in Definition~\ref{flatGrammars}.
Pyramids have Dyck labels ${\bf a}^{k} \, {\bf a}^{-k}$, for $k \geq 1$.
Two adjoint pyramids in a Dyck grid paths share a valley. This shared valley
is not a Dyck path on its own.
Indeed, pyramids are building blocks of Dyck paths in grid graphs.
In semi-Dyck grid paths, the valleys are themselves semi-Dyck words.
That is, pyramids and valleys are building blocks for semi-Dyck paths in grid graphs.
Semi-Dyck path valleys are labeled ${\bf a}^{-k} \, {\bf a}^{k}$, for $k \geq 1$.

%
% could use pyramids and canyons - peaks and valleys for top/bottom
%

\begin{definition}[Pyramids and Valleys]
{\sf
Consider an LDG $G = (\Sigma, V, E_1)$.
  A Dyck pyramid is a maximal path labeled by ${\bf a}^{k} \, {\bf a}^{-k}$, for $k \geq 1$.  
  A Dyck pyramid path $p$ is maximal since its label is ${\bf a}^{k} \, {\bf a}^{-k}$ and ${\bf a}^{k+1} \, {\bf a}^{-k-1}$ does not label a valid Dyck
  path containing $p$. 

  Semi-Dyck paths also includes maximal valleys labeled by ${\bf a}^{-k} \, {\bf a}^{k}$.
}
\label{pyramids}
\end{definition}

A {\em peak} is labeled ${\bf a} \, {\bf a}^{-1}$.
In a grid graph, a pyramid starting from $(i,j)$ and ending at $(i+t,j)$ has {\em base level} $j$.
Figure~\ref{non-flat-levels} shows base levels $0,1$ and $2$. 
The $y$-axis of Figure~\ref{gridGraph} shows grid levels.

In a grid graph, a semi-Dyck path starts from point $(0,0)$ and ends at some point $(x,0)$ for an integer $x \geq 0$.
A Dyck path in a grid graph never has a $y$ coordinate below $0$.
In general, an LDG $+1$ edge is equivalent to a grid graph edge going from $(x,y)$ to $(x+1,y+1)$, for $x \geq 0$.
Likewise, an LDG $-1$ edge is equivalent to a grid graph edge going from $(x,y)$ to $(x+1,y-1)$.

Let $p$ be an exact $0$ cost path in an LDG. In a grid graph $p$ is,
\begin{eqnarray*}
p & = & (x_1,y_1), \, (x_2,y_2), \ \cdots, \ (x_n,y_n).
\end{eqnarray*}
so that $(x_1,y_1) \, = \, (0,0)$, $(x_n,y_n) \, = \, (x_n,0)$, and $y_1 + \cdots + y_{n} = 0$.
In such a path, its {\em maximal peak(s)} are at level $\max\{ \, y_1, \ \cdots, \ y_n \, \}$.

\begin{figure}[ht]
 \vspace{-0.25cm}
 \begin{center}
 \includegraphics[height=5cm,width=16cm]{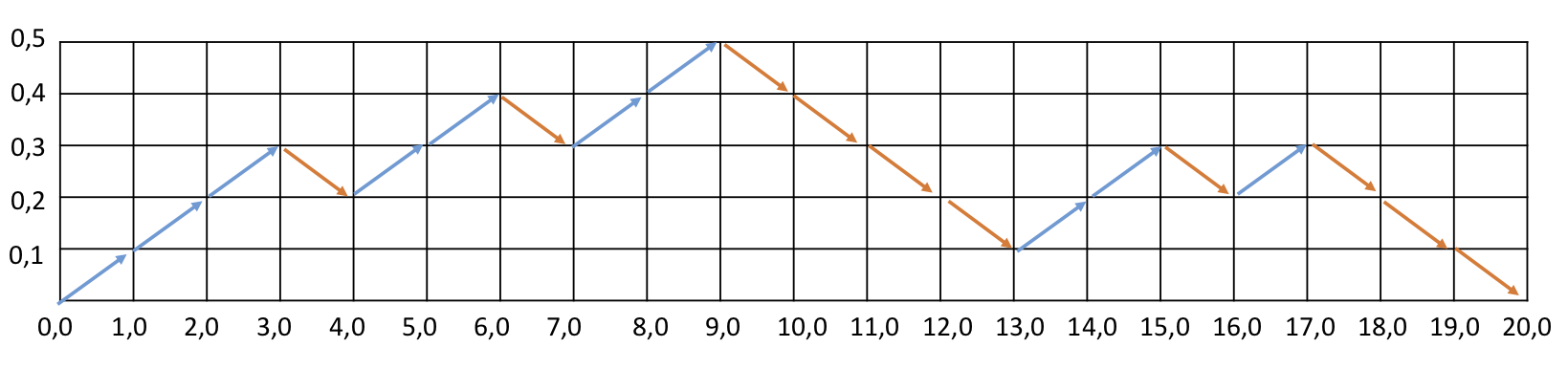}
 \end{center}
 \vspace{-0.6cm}
\caption{An exact $0$ cost Dyck path on a grid with maximum peak $(9,5)$ at level $5$}
\label{gridGraph}
\end{figure}

Start with an exact $0$ cost path $p$, then two exact $0$ cost subpaths $p_1$ and $p_2$ are distinct
iff $E[p_1] \cap E[p_2] = \emptyset$.
Suppose $p$ does not form a cycle.
Two exact $0$ cost distinct subpaths $p_1$ and $p_2$ are {\em adjoining} when they share exactly one vertex and have the same base level.
This common vertex joins the end of one of these paths to the start of the other.

\begin{definition}[Pairs]
{\sf
  A {\em pyramid pair} is an adjoining pair of pyramids.
  A {\em valley pair} is an adjoining pair of valleys.
  Likewise, a {\em mixed pair} is an adjoining pyramid (valley) and valley (pyramid).
}
\label{pyramid pairs}
\end{definition}

Figure~\ref{flatPath} shows three pyramid pairs in a flat Dyck path.
The two left peaks in Figure~\ref{non-flat-levels} form a pyramid pair on level~2, but not level~1 or~0.

\begin{figure}[ht]
 \begin{center}
 \includegraphics[height=4cm,width=12cm]{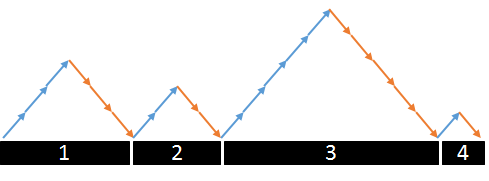}
 \end{center}
 \vspace{-0.6cm}
\caption{Three pyramid pairs $\{ \, (1,2), \, (2,3), \, (3,4) \}$ in a flat Dyck path}
\label{flatPath}
\end{figure}

Consider an exact $0$ cost path forming a pyramid pair.
Intuitively, each of these pyramids are independent since {\bf Digraph-flat-exact-paths}
finds their exact $0$ paths independently.

\subsection{The general case}
\label{general-case-subsection}

Consider a Dyck pyramid pair ${\bf a}^{k_1} \, {\bf a}^{-k_1} \, {\bf a}^{k_2} \, {\bf a}^{-k_2}$ for integers $k_1 \geq 1,  \, k_2 \geq 1$.
The word 
$${\bf a}^{s} \, {\bf a}^{k_1} \, {\bf a}^{-k_1} \, {\bf a}^{k_2} \, {\bf a}^{-k_2} \, {\bf a}^{-s},$$
has the {\em exterior pair} ${\bf a}^{s}$ and ${\bf a}^{-s}$, for an integer $s \geq 1$, see for example~\cite{DENISE1995155}.
Exterior pairs are always made of pairs of matching elements.
Semi-Dyck words have exterior pairs ${\bf a}^{s}$ and ${\bf a}^{-s}$ for any integer $s \neq 0$.
Combining exterior pairs with flat grammars gives the general Dyck and semi-Dyck cases.
%
%That is, extending Defintion~\ref{flatGrammars} to include exterior pairs gives the general Dyck and semi-Dyck grammars.
%

An {\em isolated} pyramid (valley) has no adjoining pyramid (valley).
Pyramids and valleys are isolated by exterior pairs.
If ${\bf a}^{k} \, {\bf a}^{-k}$ is an isolated pyramid, then
it is enclosed by at least one exterior pair.

The rightmost pyramid in Figure~\ref{non-flat-levels} is an isolated pyramid. This isolated pyramid has label ${\bf a}^2 \, {\bf a}^{-2}$.
There is an isolated pyramid pair at base level~2. These two pyramids are on the left.

\begin{definition}[Isolated paths]
  {\sf
  In a grid graph,
  an $m$ {\em isolated path} is any maximal sequence of $m$ pyramids or valleys all adjoining at the same base level.
}
\end{definition}

A consequence of the definition of an $m$ isolated path is it has no $(m+1)$st adjoining pyramid or valley.
Isolated paths with $m=2$ are {\em isolated pairs}.
Isolated paths with $m=4$ are {\em isolated quads}.

The four boxes in Figure~\ref{mountain} are isolated paths.
There are also two pyramid pairs with three peaks in the middle.
These pyramids are contained by an exterior pair.

An {\em invocation} of {\bf Digraph-flat-exact-paths} runs $\lceil \log n \rceil$ iterations
from line~4 in Figure~\ref{New-Min-Plus}.
One invocation of {\bf Digraph-flat-exact-paths} converts all of these isolated paths into exact $0$ cost edges.

\begin{figure}[ht]
 \begin{center}
 \includegraphics[height=4cm,width=12cm]{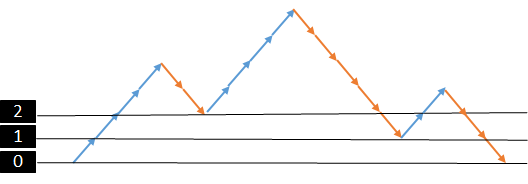}
 \end{center}
 \vspace{-0.6cm}
\caption{Level 2 has one isolated pyramid pair on the left, level 1 has an isolated pyramid on the right}
\label{non-flat-levels}
\end{figure}

\begin{figure}[ht]
 \begin{center}
 \includegraphics[height=4cm,width=12cm]{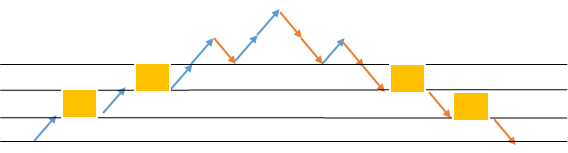}
 \end{center}
 \vspace{-0.6cm}
\caption{A Dyck or semi-Dyck path with four isolated paths in the boxes}
\label{mountain}
\end{figure}

\begin{corollary} 
{\sf
 Given an LDG $G_1 = (\Sigma, V, E_1)$, where $\Sigma$ is Dyck (semi-Dyck) and $|\Sigma| = 2$,
  then one invocation of {\bf Digraph-flat-exact-paths} finds an exact $0$ cost path for
  an $m$ isolated path, where $m: n/2 \geq m \geq 2$.
}
\label{m-pyramids}
\end{corollary}

\begin{proof}
Without loss, the focus is on an $m$ isolated path of pyramids.
Since the $m \leq n/2$ pyramids form an adjoining isolated path, they all have the same base level.
This means the $m$ pyramids form a flat Dyck path.
So an invocation {\bf Digraph-flat-exact-paths} computes the flat Dyck reachability
of all adjoining pyramid pairs by Lemma~\ref{powers-of-two}.
\end{proof}

A Dyck {\em inner segment} is a Dyck word $w$ between a set of exterior pairs ${\bf a}^s$ and ${\bf a}^{-s}$, for integers $s \geq 1$.
The word 
${\bf a}^s \, w \, {\bf a}^{-s}$ labels a valid Dyck path $p$ and $s$ is maximal so ${\bf a}^{s+1} \, w \, {\bf a}^{-s-1}$ does not label a valid
path containing $p$.
The semi-Dyck case has maximal exterior pairs ${\bf a}^{s}$ and ${\bf a}^{-s}$, for integers $s \neq 1$.
A semi-Dyck inner segment is a semi-Dyck word.

Inner segments are contained by {\em pyramid (valley) bases} in the Dyck (semi-Dyck) case.
If an inner segment is empty, then the exterior pairs ${\bf a}^s$ and ${\bf a}^{-s}$ (${\bf a}^{-s}$ and ${\bf a}^{s}$) form a pyramid (valley), for $s \geq 1$.

\begin{corollary}
{\sf
 Given an LDG $G_{\ell} = (\Sigma, V, E_{\ell})$, where $\Sigma$ is Dyck (semi-Dyck) and $|\Sigma| = 2$ and $\ell \geq 2$,
  and suppose an $0$ cost edge connects an inner segment of $s$ exterior pairs.
  Then one invocation of {\bf Digraph-flat-exact-paths} 
  finds an exact $0$ path through this inner segment and these exterior pairs.
}
\label{enclosing}
\end{corollary}

\begin{proof}
Suppose a set of $s$ exterior pairs contains a $0$ cost edge. 
It must be that, $s < n/2$, and by assumption {\bf Digraph-flat-exact-paths} already found inner segment's exact $0$ cost path
by iteration $\ell$.
Thus, {\bf Digraph-flat-exact-paths} finds the exact $0$ cost path including these exterior pairs by
Lemma~\ref{powers-of-two}.
\end{proof}

Consider {\bf Digraph-flat-exact-paths}.
Finding exact $0$ cost paths for isolated paths is not compatible with finding paths in
their exterior pairs.
This incompatibility is handled by careful iterations of {\bf Digraph-flat-exact-paths}.
See a single iteration in Figure~\ref{general-case}.
When needed, assume lines~3 and~4 are adapted for the Dyck case, see the discussion accompanying Figure~\ref{New-Min-Plus-Dyck}.

\begin{convention}[Atomic invocations of {\bf Digraph-flat-exact-paths}]
{\sf
For the worst case of Figure~\ref{general-case},
the algorithm {\bf Digraph-flat-exact-paths} is atomic.
}
\label{invocations}
\end{convention}

\begin{figure}[ht]
\begin{center}
  \leavevmode
%% put framebox and vbox on the same line!  Otherwise LaTeX complains!
\framebox[1.05\width][l]{\vbox{
\begin{tabbing}
a \= text \= text \= text \= text \= text \kill
\> 1. $M \leftarrow \mbox{\bf Digraph-flat-exact-paths}($G$)$ \ // Fig.~\ref{New-Min-Plus} (semi-Dyck add Fig.~\ref{New-Min-Plus-Dyck})\\
\> 2. $M \leftarrow M + \mbox{\bf AdjMatrix}(G)$ \ // original edges plus exact $0$ paths\\
\> 3. $M \leftarrow M^2$ \ // AGMY multiplication, extending $\pm 1$ edges with exact $0$ paths\\
\> 4. $M \leftarrow \mbox{\em Normalize\_and\_Divide\_by\_2}(M)$
\end{tabbing}
}
}
\end{center}
\vspace{-0.125in}
\caption{A single iteration for solving the general case, the general Dyck or semi-Dyck solution iterates these four steps $\lceil \log n \rceil$ times.}
\label{general-case}
\end{figure}

One invocation of {\bf Digraph-flat-exact-paths}, Figure~\ref{New-Min-Plus}, does not find the exact $0$ path from start to end
of the top path in Figure~\ref{NonFlatCases}.
Similarly, one invocation of semi-Dyck version of {\bf Digraph-flat-exact-paths} does not find the exact $0$ cost path from start to end for the bottom path.

\begin{figure}[ht]
 \begin{center}
 \includegraphics[height=3cm,width=6cm]{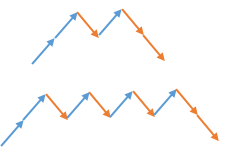}
 \end{center}
 \vspace{-0.6cm}
\caption{Cases where {\bf Digraph-flat-exact-paths} alone does not work}
\label{NonFlatCases}
\end{figure}

The next discussion illustrates the challenge of atomic invocations of {\bf Digraph-flat-exact-paths}.
Then we show how iterations of Figure~\ref{general-case} correct for these challenges.

Using label-costs and node names from $0$ to $6$, the top path in Figure~\ref{NonFlatCases} is,
\begin{eqnarray*}
p \ \ = \ \ 0 \stackrel{+1}{\longrightarrow} 1 \stackrel{+1}{\longrightarrow}  2 \stackrel{-1}{\longrightarrow} 3 \stackrel{+1}{\longrightarrow}  4 \stackrel{-1}{\longrightarrow} 5 \stackrel{-1}{\longrightarrow} 6.
\end{eqnarray*}

The Dyck algorithm {\bf Digraph-flat-exact-paths} fails to find the exact $0$ path for $p$. 
To see this, 
let $p_1$ be the subpath made of the middle two edges, from node~2 to~4. So, $p_1$ is labeled with ${\bf a}^{-1} \ {\bf a}$ and
$p_1$ is not Dyck. 
In its first iteration, the algorithm finds $\pm 2$ edges $0 \stackrel{+2}{\longrightarrow} 2$ and $4 \stackrel{-2}{\longrightarrow} 6$.
Also the $\pm 2$ edges are normalized to $\pm 1$ edges: $0 \stackrel{+1}{\longrightarrow} 2$ and $4 \stackrel{-1}{\longrightarrow} 6$.
In the Dyck case, there are no exact $0$ paths to extend these new $\pm 1$ edges, so they are removed in the second iteration.
This leaves no exact $0$ path from~$0$ to~$6$. 

Consider the path $p$ at the top of Figure~\ref{NonFlatCases}.
The semi-Dyck version of {\bf Digraph-flat-exact-paths} finds the exact $0$ cost path along $p$.
This is because the middle two edges, from node~$2$ to~$4$, are labeled ${\bf a}^{-1} \, {\bf a}$.
So in the first iteration, an exact $0$ semi-Dyck path is found from $2$ to $4$.
Also in this iteration, as in the Dyck case, a new $+1$ is created from $0$ to $2$.
Likewise, a $-1$ edge is created from $4$ to $6$.
All told, in line~9 of the first iteration, the new $+ 1$ edge is extended from node $0$ to $4$ and the $-1$ edge is extended to go
from $2$ to $6$.
Thus, the second iteration finds the semi-Dyck path from $0$ to $6$.

The flat semi-Dyck (Dyck) algorithm {\bf Digraph-flat-exact-paths} does not find the exact $0$ path for the bottom path of Figure~\ref{NonFlatCases}.
This is because in its first iteration it does not find an exact $0$ cost path from the first pyramid peak to the fourth pyramid peak.
It does find the exact $0$ path joining the three exact $0$ cost valleys using line~9.
This same iteration extends the new $\pm 1$ edges after normalizing the $\pm 2$ edges at the start and end.
The new $+1$ edge is extended to the second pyramid peak. The $-1$ edge is extended from the third pyramid peak.
These extensions are all done by line~9. So they are computed in the same algebraic matrix multiplication that forms the exact $0$ path joining the three valleys.
Therefore, the extended $\pm 1$ edges cannot reach each other, so in the next iteration they cannot form an exact $0$ cost path in line~5.

\paragraph{An iterative solution.}

The general solution is based on iterating invocations of {\bf Digraph-flat-exact-paths}.
See Figure~\ref{general-case}.
After each run of the flat path algorithm, all original $\pm 1$ edges are extended by the new exact $0$ paths.
Any new $\pm 2$ edges are normalized to $\pm 1$ edges in preparation for the next iteration.
The process is repeated for finding more exact $0$ paths.
After each invocation of {\bf Digraph-flat-exact-paths}, all new exact $0$ paths remain for subsequent iterations.

\begin{figure}[ht]
 \begin{center}
 \includegraphics[height=3cm,width=12cm]{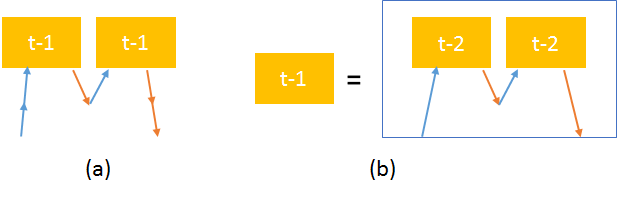}
 \end{center}
 \vspace{-0.6cm}
\caption{A worst case Dyck path for $t > 2$, where the $t=1$ case is the empty block attaching its input edge directly to its output edge, giving $m=2$ pyramids}
\label{worstCaseDyckPaths}
\end{figure}

%
%This is minimal in the sense larger nested flat paths would also be found by a single application of the 
%flat semi-Dyck path algorithm.
%

In Figure~\ref{worstCaseDyckPaths}(a), if the leftmost $t-1$ block is empty, the rightmost $t-1$
block continues recursively
using the Figure~\ref{worstCaseDyckPaths}(b), then the leftmost side is an isolated pyramid.
In general, if the leftmost (rightmost) $t-1$ block is an isolated path, then it will be replaced by an exact $0$ cost path in one invocation
of {\bf Digraph-flat-exact-paths} by Corollary~\ref{m-pyramids}.
Each time this algorithm find an exact $0$ cost path, a new $0$ cost edges is created.

\begin{figure}[ht]
 \begin{center}
 \includegraphics[height=3cm,width=15cm]{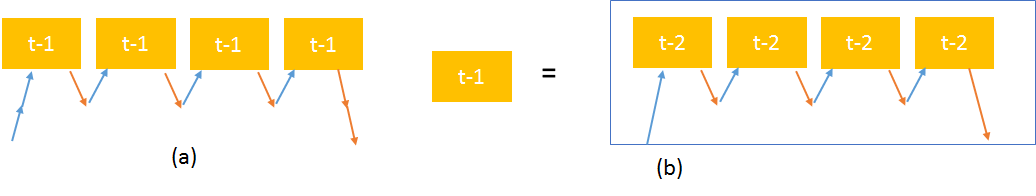}
 \end{center}
 \vspace{-0.6cm}
\caption{A worst case semi-Dyck path for $t > 2$, where the $t=1$ case is the empty block attaching its input edge directly to its output edge giving $m=4$ pyramids}
\label{worstCaseSemiDyckPaths}
\end{figure}

This is because, for Dyck paths, Corollary~\ref{m-pyramids} shows $m > 2$ pairs of adjoining pyramids require a single invocation
of {\bf Digraph-flat-exact-paths}.
Just the same, a single invocation of {\bf Digraph-flat-exact-paths} also finds an exact $0$ path for pyramid pairs.
Indeed, reducing $m>2$ adjoining pyramids to $m=2$ pyramids, frees vertices to contribute to worst case subpaths.
Similarly, 
Corollary~\ref{enclosing} indicates exterior pairs ${\bf a}^{s}$ and ${\bf a}^{-s}$, for $s \geq 1$, may be reduced to exterior pairs
${\bf a}$ and ${\bf a}^{-1}$ where $s=1$.
Since if $s>1$ and say an exact $0$ path is known for an inner segment, then 
one invocation of {\bf Digraph-flat-exact-paths} finds the exact $0$ path through these exterior pairs.
Likewise, a single invocation finds the exact $0$ path through exterior pairs ${\bf a}$ and ${\bf a}^{-1}$ when the exact
$0$ path is known for the inner segment.
So, a reduced worst case Dyck path has all exterior pairs with $s=1$.

Reduced semi-Dyck paths have exterior pairs
${\bf a}^s$ and ${\bf a}^{-s}$, for $s \in \{ \, -1, +1 \, \}$.
Reduced valleys are labeled ${\bf a}^{-1} \, {\bf a}$.
Finally, any $m > 4$ pairs of adjoining pyramids and/or valleys are replaced with $m = 4$ quads.

Worst case semi-Dyck paths are the exact $0$ paths given in Figure~\ref{worstCaseSemiDyckPaths}.
If there are only three pyramids in Figure~\ref{worstCaseSemiDyckPaths}(a), then they have two adjoining shared valleys.
The first iteration of line~5 creates a $+2$ edge at the start.
Likewise, the first iteration creates a $-2$ edge at the end.   These edges are normalized to new $\pm 1$ edges in line~7.
Finally, line~9 extends the new $\pm 1$ edges using the two exact $0$ paths just made from the two valleys. 
This gives a $+1$ edge adjoining a $-1$ edge for the next iteration.
The next iteration of {\bf Digraph-flat-exact-paths} combines these edges producing an exact $0$ path from the start to the end of this path.
In summary, semi-Dyck isolated paths with $m \leq 3$ pyramids and/or valleys become exact $0$ paths in a single invocation of {\bf Digraph-flat-exact-paths}.

\begin{figure}[ht]
 \begin{center}
  \includegraphics[height=2cm,width=6cm]{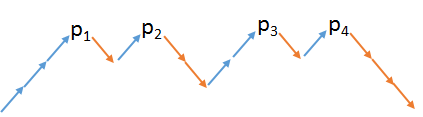}
 \end{center}
  \vspace{-0.3cm}
  \caption{In the Dyck case, if the algorithm in Figure~\ref{general-case} converts the paths $p_1, p_2, p_3$ and $p_4$ into $0$ cost edges, then these are two adjoining augmented pairs contained
  in an exterior pair}
\label{components}
\end{figure}

\begin{definition}[Augmented pairs or quads]
{\sf
  A Dyck (semi-Dyck) {\em augmented pair (quad)} is a path of two (four) adjoining pyramid (and/or valley) bases whose inner segments are $0$ cost edges.
  All adjoining bases are the same base level and all adjoining bases
  are contained by an exterior pair.
}
\end{definition}

If two augmented pairs adjoin at the same base level, then these augmented pairs become 
a single augmented pair in one iteration of Figure~\ref{general-case}. 
Another iteration replaces this single augmented pair with a $0$ cost edge.
In general, say $G'$ is a reduced Dyck grid graph with $m$ augmented pairs all adjoining at the same base level.
If all augmented pairs in $G'$ have the same maximum level peaks, then all  augmented pairs become
exact $0$ cost paths in the same iteration of Figure~\ref{general-case}.

Figure~\ref{worstCaseDyckPaths} shows how augmented pairs may be constructed.
In particular, each level of augmented pairs is contained by an exterior pair.
A new augmented pair may only form if this augmented pair has an adjoint augmented pyramid and both together are contained by an exterior pair.
The semi-Dyck augmented quads are similar, see Figure~\ref{worstCaseSemiDyckPaths}.

\begin{figure}[ht]
 \begin{center}
 \includegraphics[height=3cm,width=12cm]{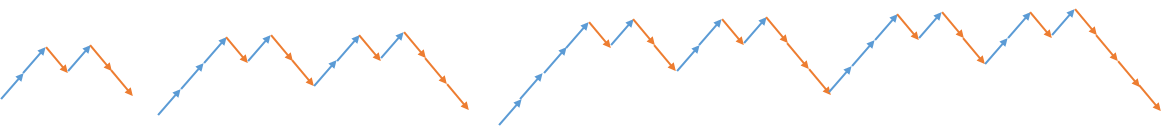} %  {bijection.png}
 \end{center}
 \vspace{-0.6cm}
\caption{Three reduced worst case Dyck paths}
\label{SomeWorstCaseDyckPaths}
\end{figure}

Figure~\ref{SomeWorstCaseDyckPaths} gives reduced Dyck path examples with $m \in \{ \ 2,4,8 \ \}$ peaks.
In general if $t(m)$ is the number of nodes $n$ in a reduced Dyck path with $m$ peaks, then
$t(m) = 2 \, t(m/2) +1$ with the base case $t(1) = 3$.
This means $t(2) = 7$, $t(4) = 15$ and $t(8) = 31$. Or if $m = 2^k$, then $n = 2^{k+2} -1$.

\begin{lemma}
{\sf
Consider a reduced Dyck grid path with $m=2^{k}$ pyramid peaks and $n=2^{k+2}-1$ nodes, for an integer $k \geq 1$.
Suppose one pyramid peak is replaced by a $0$ cost edge, then 
Figure~\ref{general-case} requires fewer than $\log m$ iterations to find the exact $0$ paths.
}
\label{imbalance}
\end{lemma}

\begin{proof}
In a grid graph let $p$ be a reduced Dyck path with $m -1$ pyramid peaks, for $m = 2^k$ where $k \geq 1$.
These peaks are all at maximum level $m$.  
Suppose, for the sake of a contradiction, determining that this path $p$ is an exact $0$ cost path
requires $\log m$ iterations.

Since there are $m = 2^{k} -1$ pyramid peaks at level $m$, there is at least one pyramid peak that is not in an isolated pair.
This single isolated pyramid becomes an exact $0$ cost path in the first iteration of Figure~\ref{general-case}.
So this exact $0$ cost path cannot contribute to creating a new augmented pair by an invocation of {\bf Digraph-flat-exact-paths}.
Therefore, at the start of the second iteration of Figure~\ref{general-case}, there is an augmented pair that does not have an
equivalent neighboring pair to form a new augmented pair.
Hence, in the third iteration, there is a new augmented pair just created from two augmented pairs that will not have an equivalent 
neighboring pair to form a new augmented pair.

In general, the after $\log m -1$ iterations of this process, there are at least
$$m \sum_{i=1}^{\log m -1} \frac{1}{2^i}$$
pyramid peaks that started at the maximum level $m$ but they could not form a new augmented pair.
This sum is larger than $\frac{1}{2}$ for $m \geq 4$. For $m=2$ peaks, dropping one peak gives a path with one peak which requires 
one invocation of the algorithm in Figure~\ref{general-case}.

In summary, this means the remaining peaks already formed an exact $0$ cost path.
Thus, giving a contradiction since $\log m -1$ iterations are sufficient to find this exact $0$ cost path.
\end{proof}

In general suppose $t_s(m)$ is the number of nodes $n$ in a reduced semi-Dyck path with $m = 4^{k}$ peaks, for $k \geq 0$.
Therefore, 
$t_s(m) = 4t_s(m/4) -1$ with the base case $t_s(1) = 3$.
This means $t_s(4) = 11$ and $t_s(16) = 43$. Or if $m = 4^k$, then $n = \frac{2^{2k+1}+1}{3}$, see~\cite[A007583]{S}.

A proof of the next lemma is similar to the proof of Lemma~\ref{imbalance}.

\begin{lemma}
{\sf
Consider a reduced semi-Dyck grid path with $m=4^{k}$ pyramids or valleys and $n=\frac{2^{2k+1}+1}{3}$ nodes, for an integer $k \geq 1$.
Suppose one pyramid or valley is replaced by an exact $0$ path, then 
Figure~\ref{general-case} requires fewer than $\log_4 m$ iterations.
}
\label{imbalance2}
\end{lemma}

Lemmas~\ref{imbalance} and~\ref{imbalance2} show the worst case Dyck or semi-Dyck paths have maximum peaks at height $O(\log m)$.
Moreover, these results indicate it is sufficient to consider $m$ peaks when $m$ is a power of~2 for the Dyck case and powers of~4 for the semi-Dyck case.

\begin{lemma}
{\sf
  Given an LDG $G = (\Sigma, V, E_1)$, where $\Sigma$ is Dyck (semi-Dyck) and $|\Sigma| = 2$, then a worst case path for
  iterations of Figure~\ref{general-case} is given by Figure~\ref{worstCaseDyckPaths} (Figure~\ref{worstCaseSemiDyckPaths}).
}
\label{worst-case}
\end{lemma}

\begin{proof}
Without loss, consider only reduced Dyck paths.
  
All Dyck paths are built from pyramids.
In reduced paths, all pyramids are in isolated pyramid pairs.
The reduced pyramids are ${\bf a} \, {\bf a}^{-1}$.
Furthermore, all isolated pyramid pairs have an exterior pair ${\bf a}^s$ and ${\bf a}^{-s}$ where $s=1$.

By Lemma~\ref{imbalance}, the reduced worst case must start with $m = 2^{k}$ pyramid peaks, for $k \geq 1$, at maximum level $m$.
The algorithm in Figure~\ref{general-case} finds the exact $0$ path in such paths in $\lceil \log m \rceil \leq \lceil \log n \rceil$ iterations,
since $m \leq n$.
  
A similar argument holds in the semi-Dyck case.
The main difference is: semi-Dyck paths are built from pyramids and valleys.
Moreover, 
adjoining isolated paths may be reduced to $m=4$ adjoining elements that are isolated together.
Given these differences, all the Dyck arguments just presented remain the same.

This completes the proof.
\end{proof}

In the next lemma, height of a grid graph is the $y$ value of maximum peak $(x,y)$.

\begin{lemma}
{\sf
  Given an LDG $G = (\Sigma, V, E_1)$, where $\Sigma$ is Dyck (semi-Dyck) and $|\Sigma| = 2$, then in the worst case finding all
  exact $0$ paths takes $\lceil \log n \rceil$ iterations of 
  Figure~\ref{general-case}.
}
\label{worst-case-calculation}
\end{lemma}

\begin{proof}
Without loss, the focus is on Dyck paths.
Let $h(n)$ be the height of a worst-case Dyck path for iterations of the algorithm in Figure~\ref{general-case}.

By Lemma~\ref{worst-case}, the worst-case paths must double their number of isolated pairs at each level.
This means after the first iteration, each subsequent iteration of the algorithm in Figure~\ref{general-case} halves the number of augmented pairs.

Therefore, $h(n) \leq  h(\lceil n/2 \rceil) +1$ which immediately means $h(n) = O(\log n)$.
In particular, the additive term of~$1$ is for each invocation of 
{\bf Digraph-flat-exact-paths}.
\end{proof}

This section culminates in the main theorem.

\begin{theorem}
{\sf
  Given an LDG $G = (\Sigma, V, E_1)$, where $\Sigma$ is Dyck (semi-Dyck) and $|\Sigma| = 2$,
  then Figure~\ref{general-case}'s algorithm solves Dyck (semi-Dyck) reachability in $O(n^{\omega} \log^{3} n)$ time.
}
\label{MAIN}
\end{theorem}

\section{Determining $\pm 1$ reachability}
\label{sec:plus-minus-one}

Determining $\pm 1$ path reachability in an LDG is based on $0$ cost edges computed by the algorithm in Figure~\ref{general-case}.
After running this algorithm, each $0$ cost edge represents an exact $0$ cost reachability path. 
This reachability is either Dyck and semi-Dyck reachability.

\begin{definition}
{\sf
Let $E^{*}$ contain all $0$ cost edges found by running the algorithm in Figure~\ref{general-case} $\lceil \log n \rceil$ times.
}
\end{definition}

The focus is on the $E_1$ edges of $G_1 = (\Sigma, V, E_1)$ in combination with the exact $0$ cost edges $E^{*}$.

In the case of $\pm 1$ reachability, consider the next paths from $i$ to $j$. Both $i \longrightarrow k_1$ and $k_2 \longrightarrow j$ are edges in $E_1$
and $k_1 \stackrel{0}{\curvearrowright} k_2$ is an exact $0$ cost edge from $E^{*}$.
So, all cases for $i \stackrel{u}{\longrightarrow} k_1$ and $k_2 \stackrel{v}{\longrightarrow} j$ so that $u,v \in \{ \, -1, +1 \, \}$ are:

      \begin{eqnarray*}
        i \stackrel{+1}{\longrightarrow} k_1 \stackrel{0}{\curvearrowright} k_2 \stackrel{+1}{\longrightarrow} j,\\
        i \stackrel{+1}{\longrightarrow} k_1 \stackrel{0}{\curvearrowright} k_2 \stackrel{-1}{\longrightarrow} j,\\
        i \stackrel{-1}{\longrightarrow} k_1 \stackrel{0}{\curvearrowright} k_2 \stackrel{+1}{\longrightarrow} j,\\
        i \stackrel{-1}{\longrightarrow} k_1 \stackrel{0}{\curvearrowright} k_2 \stackrel{-1}{\longrightarrow} j.
      \end{eqnarray*}
If these are the only paths from $i$ to $j$, then there is no exact $\pm 1$ path from $i$ to $j$.

Of course, all edges in $E^{*}$ are built from edges in $E_1$.
An edge is {\em directly} from $E_1$ if it is not in an exact $0$ cost edge under discussion.

\begin{lemma}
  {\sf
  Given an LDG $G_1 = (\Sigma, V, E_1)$,
  its exact $0$ cost paths are $0$ edges in $E^{*}$ where $\Sigma$ is Dyck (semi-Dyck) and $|\Sigma| = 2$, then
  all $\pm 1$ paths can be found by extending all $\pm 1$ edges in $E_1$ with only $0$ cost edges in $E^{*}$.
}
\label{single_PM_edge}
\end{lemma}

\begin{proof}
Without loss, the semi-Dyck case is the focus. 
Recall, semi-Dyck paths are paths with the same number of $+1$ and $-1$ edges from $E_1$.

  Consider the edge $e_1 = i \stackrel{\pm 1}{\longrightarrow} k_1 \in E_1$ and another edge $e_2 = k_2 \stackrel{u}{\longrightarrow} j$ in $E_1$, where $u \in \{ \, -1, +1 \, \}$, 
  so that there is a $0$ cost edge from  $k_1$ to  $k_2$ in $E^*$. 
Let $k_1 \stackrel{0}{\curvearrowright} k_2$ be this $0$ cost edge in $E^*$. 
Suppose, for the sake of a contradiction, that $i \stackrel{\pm 1}{\longrightarrow} j$, 
is a $\pm 1$ path that is not discovered by extending $e_1$ with edges from $E^*$.
Consider the next cases.\\

\noindent
{\bf Case 1}: If $e_1$ and $e_2$ have different signs.\\

If $e_1$ and $e_2$ have different signs, then the entire path $i \stackrel{0}{\longrightarrow} j$ is another $0$ cost edge in $E^{*}$ found by
the algorithm in Figure~\ref{general-case}.
This is a contradiction, since in this case joining $e_1$ and $e_2$ does not create a $\pm 1$ edge.\\

\noindent
{\bf Case 2}: If $e_1$ and $e_2$ have the same sign.\\

Given the exact $0$ edge $k_1 \stackrel{0}{\curvearrowright} k_2$ from $E^{*}$ between
the edges $e_1$ to  $e_2$
gives a $\pm 2$ path: $p_{1,2} = \overset{\pm 1}{e_1} \stackrel{0}{\curvearrowright} \overset{\pm 1}{e_2}$.

Say the path $p_{1,2}$ contributes to a $\pm 1$ edge in combination with another opposite sign edge $e_3$.
There are two subcases that both lead to contraditions.\\

\begin{myindentpar}{1cm}
{\bf Subcase 2a}: $\overset{\pm 1}{e_1} \stackrel{0}{\curvearrowright} \overset{\pm 1}{e_2} \stackrel{0}{\curvearrowright} \overset{\mp 1}{e_3}$.\\

This subcase gives $p_{2,3} = \overset{\pm 1}{e_2} \stackrel{0}{\curvearrowright} \overset{\mp 1}{e_3}$ which has exact cost $0$.
Thus, $p_{2,3}$ must be and edge in $E^*$. 
Therefore, the $\pm 1$ reachabilty of $\overset{\pm 1}{e_1} \stackrel{0}{\curvearrowright} \overset{\pm 1}{e_2} \stackrel{0}{\curvearrowright} \overset{\mp 1}{e_3}$ requires only $e_1 = i \stackrel{\pm 1}{\longrightarrow} k_1$ to be directly from $E_1$.\\

{\bf Subcase 2b}: $\overset{\mp 1}{e_3} \stackrel{0}{\curvearrowright} \overset{\pm 1}{e_1} \stackrel{0}{\curvearrowright} \overset{\pm 1}{e_2}$.\\

This subcase gives $p_{3,1} = \overset{\mp 1}{e_3} \stackrel{0}{\curvearrowright}  \overset{\pm 1}{e_1}$ which also has cost $0$.
Thus, $p_{3,1}$ must be in $E^*$.
This means the $\pm 1$ reachabilty of $\overset{\mp 1}{e_3} \stackrel{0}{\curvearrowright} \overset{\pm 1}{e_1} \stackrel{0}{\curvearrowright} \overset{\mp 1}{e_2}$ requires only $e_2 = k_2 \stackrel{\pm 1}{\longrightarrow} j$ to be directly from $E_1$.\\

\end{myindentpar}

Since Dyck languages are also semi-Dyck, the proof is complete.
\end{proof}

Finding all exact $0$ cost edges in an LDG $G$ with all edges initially labeled from $\{ \, -1, 1 \, \}$ is the same as determining all $0$ reachability in $G$.

Consider the output $E^*$ from Figure~\ref{general-case}.
Lemma~\ref{single_PM_edge} indicates finding all exact $\pm 1$ cost paths may be computed as 
the AGMY matrix product,

\begin{eqnarray*}
E^{*} \, M \, E^{*}
\end{eqnarray*}

\noindent
where $M$ is the adjacency matrix of the given LDG $G = (V,E_1)$.
This costs $\widetilde{O}(n^{\omega})$.

\section{Conclusion}

Combining Theorem~\ref{MAIN} with Lemma~\ref{single_PM_edge} gives useful results for finding
exact Dyck and semi-Dyck paths in digraphs.

Starting with an LDG $G = (\Sigma, V, E_1)$, where $\Sigma$ is Dyck with $|\Sigma| = 2$,
then all $\{ -1, 0, +1 \}$~cost edges can be found in
in $O(n^{\omega} \log^3 n)$ time.
A number of powerful techniques can reduce this cost by polylog factors giving $\widetilde{O}(n^{\omega})$, see~\cite{AGM,AHU:1974,C,Zwick:2002,Rytter}.
In particular, Zwick~\cite{Zwick:2002} outlines such improvements nicely.

\section*{Acknowledgments}

Thanks to Emily Proulx and Sarthak Behl for contributing example inputs and discussing applications of these results.

Thanks to Derek Morris for pointing out interesting applications as well as his encouragement.

\bibliographystyle{IEEEtran/IEEEtran}
\bibliography{IEEEabrv,pathP}

\end{document}